\documentclass[12pt]{iopart}
\usepackage{setstack}
\usepackage{amsthm}
\usepackage{amsfonts}
\usepackage{amssymb}
\usepackage{graphicx}
\usepackage{subfig}
\usepackage{bm}

\newcommand{\ud}{\,\textrm{d}}
\newcommand{\RR}{\mathbb{R}}

\newcommand{\bxi}{\boldsymbol{\xi}}
\newcommand{\bx}{\boldsymbol{x}}
\newcommand{\by}{\boldsymbol{y}}
\newcommand{\balpha}{\boldsymbol{\alpha}}
\newcommand{\bz}{\boldsymbol{z}}
\newcommand{\bv}{\boldsymbol{v}}
\newcommand{\bu}{\boldsymbol{u}}
\newcommand{\be}{\boldsymbol{e}}
\newcommand{\FF}[1]{{}_0\mathcal{F}_1^{(#1)}}
\newcommand{\FZ}[1]{{}_0\mathcal{F}_0^{(#1)}}

\newcommand{\PP}[2]{\mathcal{P}_{#1}^{(#2)}}

\newcommand{\bone}{\boldsymbol{1}}
\newcommand{\bzero}{\boldsymbol{0}}
\newcommand{\bzeta}{\boldsymbol{\zeta}}
\newcommand{\bnabla}{\boldsymbol{\nabla}}

\newcommand{\bmm}{\boldsymbol{m}}

\newcommand{\HH}[1]{\mathcal{H}_{\text{#1}}}
\newcommand{\spn}{\text{Span}}

\newcommand{\bos}{\boldsymbol{s}}
\newcommand{\bor}{\boldsymbol{r}}
\newcommand{\bX}{\boldsymbol{X}}
\newcommand{\bY}{\boldsymbol{Y}}
\newcommand{\blambda}{\boldsymbol{\lambda}}
\newtheorem{theorem}{Theorem}
\newtheorem{lemma}[theorem]{Lemma}
\newtheorem{proposition}[theorem]{Proposition}

\graphicspath{{figures/}}

\begin{document}

\title{Two limiting regimes of interacting Bessel processes}

\author{Sergio Andraus$^1$, Makoto Katori$^2$ and Seiji Miyashita$^1$}
\address{$^1$ Department of Physics, Graduate School of Science, University of Tokyo,
7-3-1 Hongo, Bunkyo-ku, Tokyo 113-0033}
\address{$^2$ Department of Physics, Graduate School of Science and Engineering, Chuo University, 1-13-27 Kasuga, Bunkyo-ku, Tokyo 112-8551}
\ead{andraus@spin.phys.s.u-tokyo.ac.jp}

\begin{abstract}

We consider the interacting Bessel processes, a family of multiple-particle systems in one dimension where particles evolve as individual Bessel processes and repel each other via a log-potential. We consider two limiting regimes for this family on its two main parameters: the inverse temperature $\beta$ and the Bessel index $\nu$. We obtain the time-scaled steady-state distributions of the processes for the cases where $\beta$ or $\nu$ are large but finite. In particular, for large $\beta$ we show that the steady-state distribution of the system corresponds to the eigenvalue distribution of the $\beta$-Laguerre ensembles of random matrices. We also estimate the relaxation time to the steady state in both cases. We find that in the freezing regime $\beta\to\infty$, the scaled final positions of the particles are locked at the square root of the zeroes of the Laguerre polynomial of parameter $\nu-1/2$ for any initial configuration, while in the regime $\nu\to\infty$, we prove that the scaled final positions of the particles converge to a single point. In order to obtain our results, we use the theory of Dunkl operators, in particular the intertwining operator of type $B$. We derive a previously unknown expression for this operator and study its behaviour in both limiting regimes. By using these limiting forms of the intertwining operator, we derive the steady-state distributions, the estimations of the relaxation times and the limiting behaviour of the processes.

\end{abstract}



\submitto{\JPA}

\maketitle

\section{Introduction}\label{intro}

We consider a family of interacting particle systems in $\RR_+=\{x\in\RR: x\geq 0\}$, where a finite number $N$ of  Brownian motions (BMs) repel each other and their reflection with respect to the origin, while receiving a repulsive force from the origin as well. We denote by  $p(t,\by|\bx)$ the transition probability density (TPD) of the particles in the processes arriving at the positions $\by=(y_1,\ldots,y_N)\in\RR_+^N$ from the initial positions $\bx=(x_1,\ldots,x_N)\in\RR_+^N$ after a time-duration $t>0$. If we denote the Laplacian operator by $\Delta^{(x)}=\sum_{i=1}^N \partial^2/\partial x_i^2$, then the Kolmogorov backward equation (KBE) of these processes is \cite{katoritanemura04}
\begin{eqnarray}
\fl\frac{\partial}{\partial t}p(t,\by|\bx)=&&\frac{1}{2}\Delta^{(x)} p(t,\by|\bx)+\frac{\beta}{2}\Bigg[\sum_{i=1}^N\frac{2\nu+1}{2x_i}\frac{\partial}{\partial x_i}p(t,\by|\bx)\nonumber\\
&&+\!\!\!\!\!\sum_{1\leq i\neq j\leq N}\frac{1}{x_i-x_j}\frac{\partial}{\partial x_i}p(t,\by|\bx)+\!\!\!\!\!\sum_{1\leq i\neq j\leq N}\frac{1}{x_i+x_j}\frac{\partial}{\partial x_i}p(t,\by|\bx)\Bigg].\label{typebradialdunklkbe}
\end{eqnarray}
We call this particle system the {\it interacting Bessel process}. As explained below, $\beta>0$ and $\nu \geq -1/2$, and a two-parameter family of $(\beta,\nu)$-interacting Bessel processes is considered.

The first two terms of \eref{typebradialdunklkbe} describe the diffusive motion of each individual particle. In this model, an individual particle is not a simple one-dimensional BM, but a Bessel process. For $d \in \{2,3,\ldots\} $ the Bessel process is defined as the distance from the origin of a $d$-dimensional BM. By It\^o's formula, the KBE is given by
\begin{equation}
\frac{\partial}{\partial t}p_{\textrm{Bes}(d)}(t,y|x)=\frac{\partial^2}{\partial x^2}p_{\textrm{Bes}(d)}(t,y|x)+\frac{2\nu+1}{2x}p_{\textrm{Bes}(d)}(t,y|x),\label{besselkbe}
\end{equation}
with $\nu=d/2-1\geq -1/2$ \cite{karatzasshereve91}. For $d=1$, the Bessel process is defined as the absolute value of a BM and a reflection condition at the origin is assumed to solve \eref{besselkbe}. If we denote by $I_\nu(x)$ the modified Bessel function of the first kind, $p_{\textrm{Bes}(d)}(t,y|x)$ is given by
\begin{equation}
p_{\textrm{Bes}(d)}(t,y|x)=\frac{y^{\nu+1}}{x^\nu}\frac{1}{t}\rme^{-(x^2+y^2)/2t}I_\nu \Big(\frac{xy}{t}\Big).\label{besseltpd}
\end{equation}
This expression can be extended to all real values $\nu>-1/2$, and then the Bessel process with a continuous parameter $\nu$ is defined \cite{karatzasshereve91}. (Since the modified Bessel function plays an important role in its definition, this process is called the Bessel process.) The long-term behaviour of the Bessel process depends on $\nu$. If $-1/2\leq\nu<0$ the Bessel process is recurrent, while when $\nu> 0$ the process is transient. 

The interaction between particles represented by the third and fourth terms on the rhs of \eref{typebradialdunklkbe} is repulsive among the $N$ Bessel processes and their reflected positions with respect to the origin. The parameter $\beta/2$ is put on all of the drift terms (second, third and fourth terms) following the convention from random matrix theory \cite{mehta04, forrester10}. The parameter $\beta$ gives the inverse temperature if the system has a thermal equilibrium state, which will be related with the $\beta$-Laguerre ensembles of random matrices \cite{dumitriuedelman02, dumitriuedelman05}. The case $\beta=2$ of these processes has been studied as the eigenvalue process of matrix-valued BMs with chiral symmetries \cite{bru91}, as an application of the multidimensional Yamada-Watanabe theorem \cite{graczykmalecki13} and as the noncolliding Bessel process \cite{katoritanemura04, konigoconnell01}; additionally, the determinantal structures of spatio-temporal correlation functions were clarified in \cite{katoritanemura11}. An entrance law from the $\beta=1$ Laguerre ensemble (the chiral GOE eigenvalue ensemble) to the $(2,\nu)$-interacting Bessel process was studied in \cite{katori12}, in which the spatio-temporal correlation functions are described by Pfaffians \cite{katoritanemuraPTRF07}. These determinantal and Pfaffian structures of correlation are, however, not expected for $\beta\neq 2$.

The interacting Bessel processes are realized for all values $\beta>0,\ \nu\geq-1/2$ by the radial Dunkl processes of type $B$ \cite{demni08A}. A useful feature of Dunkl operators \cite{dunkl89} and Dunkl processes \cite{roslervoit98} is the existence of the intertwining operator of type $B$, $V_B$ \cite{dunkl91}, which transforms the KBE of a set of independent Brownian motions with a symmetric initial condition into \eref{typebradialdunklkbe}  (see \ref{generalreview}). Therefore, we can use $V_B$ to study the main properties of the interacting Bessel processes. In particular, $p(t,\by|\bx)$ is proportional to a multivariate special function known as the generalized Bessel function of type $B$ \cite{bakerforrester97, roslervoit08}. It is obtained by applying $V_B$ on the exponential function symmetrized with respect to the root system of type $B$ (see \ref{generalreview} and \eref{genbesseltypeb} for the exact definition). In view of \eref{besseltpd}, it is fitting that we named the processes defined by \eref{typebradialdunklkbe} as interacting Bessel processes, because each individual particle is a {\it Bessel process}, and the TPD of $N$ such interacting particles can be described by using the {\it generalized Bessel function}.

In the present paper, we study the behaviour of the interacting Bessel processes in the following two limiting regimes:
\begin{eqnarray}
&&\beta\to\infty\mbox{ with $\nu$ fixed,}\label{limitbeta}\\
&&\nu\to\infty\mbox{ with $\beta$ fixed.}\label{limitgamma}
\end{eqnarray}
We call \eref{limitbeta} the freezing regime. First, we obtain the time-scaled steady-state distribution of the processes for large but finite values of $\beta$ and $\nu$. We find that these steady-state distributions are independent of the initial distribution of the system, and we obtain an estimation of the time required to reach the steady state in both cases. In particular, the steady-state distribution when $\beta$ is large corresponds (after a variable substitution) to the eigenvalue distribution of the $\beta$-Laguerre ensembles of random matrices \cite{dumitriuedelman05}. Then, we calculate the scaled particle distribution exactly after taking both limits for an arbitrary initial distribution. Using numerical simulations, we illustrate our results by plotting the particle density of the interacting Bessel processes as $\beta$ and $\nu$ approach either limit. For the regime \eref{limitbeta}, Figure~\ref{figRDunklB}(a) (Section~\ref{TypeB}) shows that for $\beta=2$ and $\nu=1/2$ the particle density of the process relaxes to the exact distribution after a sufficiently long time. Figure~\ref{figRDunklB}(b) shows how the steady-state distribution changes as $\beta$ grows and how the distribution takes the form of a sum of delta functions centred at the square root of the zeroes of the associated Laguerre polynomials. We prove these facts in Theorem~\ref{freezingtypeb}. Similarly, in Figure~\ref{figRDunklBoneta}(a) (Section~\ref{TypeB}) we observe how the particle density relaxes to the steady-state distribution for $\nu=16$ and $\beta=2$, and we observe that our numerical and analytical results are consistent with the known exact density. In Figure~\ref{figRDunklBoneta}(b), we depict the steady-state density of the processes as $\nu$ approaches the limit \eref{limitgamma}, and we observe that a single peak forms, indicating that all particles tend to freeze around a single point. We prove these observations in Theorem~\ref{freezingoneta}. In order to derive the results of Theorems~\ref{freezingtypeb} and \ref{freezingoneta}, we give an expression for $V_B$ for the case in which it is applied on symmetric polynomials of the variables $\{x_i^2\}_{i=1}^N$, which is derived from the generalized Bessel function of type $B$. We also study the behaviour of $V_B$ in the regimes \eref{limitbeta} and \eref{limitgamma}. We find that in these limiting regimes, $V_B$ takes the form of the intertwining operator of type $A$, $V_A$, studied in our previous paper \cite{andrauskatorimiyashita12}. In addition, these asymptotic forms of $V_B$ allow us to estimate the relaxation time to the steady state distributions given in Theorems~\ref{freezingtypeb} and \ref{freezingoneta}.

This paper is organized as follows: in Section~\ref{TypeB}, we present the statements of Theorems~\ref{freezingtypeb} and \ref{freezingoneta}, and we illustrate them through numerical simulations. In Section~\ref{sectionvb}, we present the explicit form of $V_B$ for symmetric polynomials of squared variables as Proposition~\ref{vktypeb}. We use this expression for $V_B$ in Section~\ref{proofs}, where we give the proofs of Theorems~\ref{freezingtypeb} and \ref{freezingoneta}. As part of these proofs, we study the behaviour of $V_B$ and the generalized Bessel function of type $B$ in the regimes \eref{limitbeta} and \eref{limitgamma}. We consider some open problems and other concluding remarks in Section~\ref{conclusions}.

\section{Main results}\label{TypeB}

It was noted in \cite{andrauskatorimiyashita12} that Dyson's model with $\beta>0$ (an interacting Brownian motion) has a freezing limit that depends on the roots of the Hermite polynomials. Let us define the $N$-dimensional delta function as $\delta^{(N)}(\bx)=\prod_{i=1}^N \delta(x_i)$ and the TPD of the interacting Brownian motion as $p_A(t,\by|\bx)$. Suppose that the interacting Brownian motion is restricted to start from a normalized initial distribution $\mu(\bx)$ defined on the Weyl chamber of type $A$, $C_A=\{\bx\in\RR^N:x_1<x_2<\cdots<x_N\}$, and that the corresponding distribution at time $t$ is given by
\begin{equation}\label{TypeADistributionGeneral}
f_A(t,\by)\ud\by=\int_{C_A}p_A(t,\by|\bx)\mu(\bx)\ud\bx\ud\by.
\end{equation}
Let us use the notation $x^2=\bx\cdot\bx=|\bx|^2$. Then, for $\beta\gg 1$ one can write
\begin{equation}
\fl f_A(t,\sqrt{\beta t}\bu)(\beta t)^{N/2}\ud\bu\propto \rme^{-\beta u^2/2}\prod_{1\leq i<j\leq N}|u_j-u_i|^\beta \beta^{N/2} \ud\bu[1+O(t^{-\eta/4})]
\end{equation}
for $t\gg 1/\beta^2$ and any initial distribution with power-law behaviour at infinity, $\mu(\bx)\sim x^{-N-\eta}$, $\eta>0$. This equation is obtained from \eref{TypeADistributionGeneral} by substituting $\by=\sqrt{\beta t}\bu$ and by approximating the form of the integral over $\bx$ for large values of $t$. (We refer to the proof of Theorem~\ref{freezingtypeb}, Section~\ref{proofs} for details on how this expression is obtained, in particular Equations~\eref{almosttheorem1} to \eref{DynamicalExpectationConverging}.) After the substitution $u_i=\lambda_i/\sqrt \beta$, this expression corresponds to the eigenvalue distribution of the $\beta$-Hermite ensembles of random matrices (Equation~(1) of \cite{dumitriuedelman02}). Additionally, if we denote the symmetric group acting on the components of $N$-dimensional vectors by $S_N$, the particle distribution is given by the following sum of delta functions in the freezing regime,
\begin{equation}\label{freezing}
\lim_{\beta\to\infty}f_A(t,\sqrt \beta \bv)\beta^{N/2}=\sum_{\rho\in S_N} \delta^{(N)}(\bv-\sqrt{t}\rho\bz_N).\label{frozenradiala}
\end{equation}
This holds for an arbitrary initial configuration and any time-duration $t>0$. Here, $\bz_N$ is the vector whose components are the roots of the $N$th Hermite polynomial $H_N(x)$ (in increasing order). This is a stronger statement than the one presented in Theorem~4 of \cite{andrauskatorimiyashita12}, but it is readily proved using the same arguments used in the proofs of Theorems~\ref{freezingtypeb} and \ref{freezingoneta} below (see Section~\ref{proofs}). 

In Figure~\ref{FigRDunklA}, we depict the result of numerical simulations of the interacting Brownian motions for several values of $\beta$. The stochastic differential equations (SDEs) for Dyson's model with $\beta>0$ denoted by $\bX(t)=(X_1(t),X_2(t),\ldots,X_N(t)),\ t\geq0$ are given by \cite{katoritanemura04}
\begin{equation}
\ud X_i(t)=\ud B_i(t) + \frac{\beta}{2}\sum_{\substack{j:j\neq i \cr j=1}}^N\frac{\ud t}{X_i(t)-X_j(t)},\quad i=1,2,\ldots,N.\label{dysonSDEs}
\end{equation}
Here, the $\{B_i(t)\}_{i=1}^N$ are independent Brownian motions. We have integrated the SDEs \eref{dysonSDEs} numerically for seven particles and various values of $\beta$ up to time $t=1$, and we have repeated the process $10^6$ times with a step size of $5\times10^{-5}$, with an initial configuration in which the Brownian motions start from the positions $0,\pm10^{-2},\pm2\times10^{-2}$ and $\pm3\times10^{-2}$. The graph depicts the particle density with a resolution of $10^{-2}$ after scaling down the final positions of the process by a factor of $\sqrt{\beta}$. The vertical lines in the figure denote the exact value of the roots of the Hermite polynomials. As the value of $\beta$ grows, the particles' final positions approach the roots of $H_N(x)$, so the graph suggests that as $\beta\to\infty$, the probability peaks in the figure become delta functions, as expected from \eref{frozenradiala}. It is known \cite{katoritanemura04} that when the initial position of all particles is the origin, the particle density of the interacting Brownian motions at a fixed time must coincide with the density of the Gaussian random matrix ensembles after a suitable scaling. Therefore, the densities depicted in Figure~\ref{FigRDunklA} must coincide with the eigenvalue density of the $\beta$-Hermite ensembles considered by Dumitriu and Edelman \cite{dumitriuedelman02}. Indeed, our results are consistent with their low-temperature approximations of the $\beta$-Hermite ensembles of random matrices \cite{dumitriuedelman05}, where the probability peaks are given by the zeroes of the Hermite polynomials scaled down by a factor $\sqrt{2N}$. 

\begin{figure}[!h]
                \centering
                \includegraphics[width=0.7\textwidth]{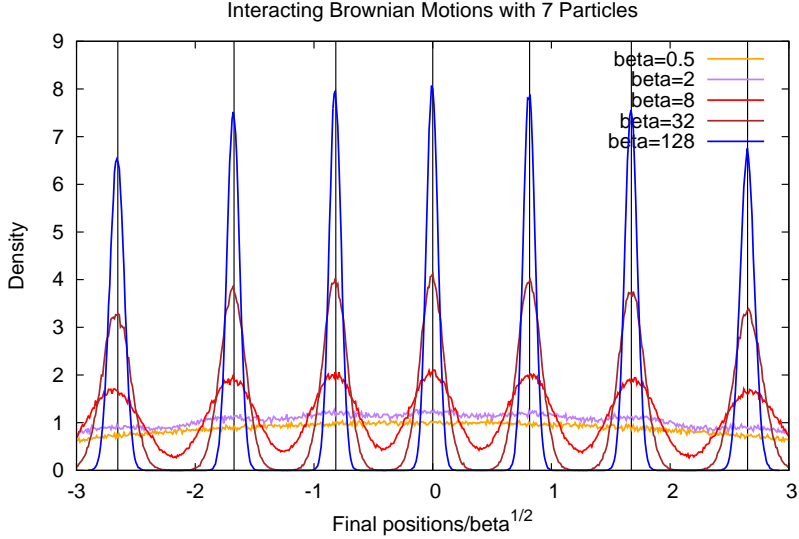}
                \caption{(Colour online) Particle density of the interacting Brownian motions at $t=1$ for several values of $\beta$. These curves correspond to the time-scaled steady-state distributions, because simulations for larger values of $t$ yield identical curves as those shown here after scaling down the final positions by a factor of $\sqrt t$.} \label{FigRDunklA}
\end{figure}

\subsection{Freezing regime of interacting Bessel processes}

A similar situation arises in the limit \eref{limitbeta} of the interacting Bessel processes. In this case, the freezing regime involves the zeroes of the the $N$th associated Laguerre polynomial $L_N^{(\alpha)}(x)$. This polynomial can be defined by the following Rodrigues' formula \cite{szego}
\begin{equation}
\rme^{-x}x^{\alpha}L_N^{(\alpha)}(x)=\frac{1}{N!}\Big(\frac{\ud}{\ud x}\Big)^N(\rme^{-x}x^{N+\alpha}).
\end{equation}
If we denote the vector of zeroes of $L_N^{(\alpha)}(x)$ in ascending order by $\bos_{\alpha}=(s_{1,\alpha},\ldots,s_{N,\alpha})$, then the vector \smash{$\bz_N=(\sqrt{s_{1,\alpha}},\ldots,\sqrt{s_{N,\alpha}})$} with $\alpha=\nu-1/2$, its permutations and sign changes make up the set of points where $f(t,\by)$ converges in the freezing regime and where the steady-state distributions attain their maxima. Consider a normalized initial distribution $\mu(\bx)$ defined in the Weyl chamber of type $B$, $C_B=\{\bx\in\RR^N:0<x_1<x_2<\cdots<x_n\}$, and the associated particle distribution,
\begin{equation}
f(t,\by)\ud\by=\int_{C_B}p(t,\by|\bx)\mu(\bx)\ud\bx\ud\by.
\end{equation}
Let us define the function
\begin{eqnarray}
\fl F(\bz,\nu,N)&=&z^2-(\nu+1/2)\sum_{i=1}^N\log z_i^2-2\sum_{1\leq i<j \leq N}\log |z_j^2-z_i^2|-N(N+\nu-1/2)\nonumber\\
\fl&&+\sum_{i=1}^N i\log i +\sum_{i=1}^N(\nu+i-1/2) \log (\nu+i-1/2),\label{functionFforb}
\end{eqnarray}
and let us denote the group composed by all the permutations and sign changes of the components of $N$-dimensional vectors by $W_B$ (see \ref{generalreview}, \eref{RootSystemTypeB} and \eref{ReflectionsTypeB} for details on the definition of $W_B$.)

\begin{theorem}\label{freezingtypeb}
Assume that $\mu(\bx)$ is a Riemann-integrable distribution with power-law decay $\mu(\bx)\sim x^{-N-\eta}$ at infinity, with $\eta>0$. Let $C$ and $C^\prime$ be large positive constants, and set $\alpha=\nu-\frac{1}{2}$. Then, for finite $\beta> C/(N+\alpha)$, the distribution $f(t,\by)$ is given by
\begin{equation}
f(t,\sqrt{\beta t}\bu)(\beta t)^{N/2}\ud\bu=c\,\rme^{-\beta F(\bu,\alpha+1/2,N)/2}\beta^{N/2}\ud\bu(1+O(t^{-\eta/4}))\label{approxbeta}
\end{equation}
with a normalization constant $c$, provided $t> C^\prime/[\beta^2 N^2(N+\alpha)^2]$. This distribution has its maxima at $\bz_N=(\sqrt{s_{1,\alpha}},\ldots,\sqrt{s_{N,\alpha}})$.
Furthermore, in the freezing regime \eref{limitbeta}, for any $\mu(\bx)$ and $t>0$ we have
\begin{equation}
\lim_{\beta\to\infty}f(t,\sqrt{\beta}\bv)\beta^{N/2}\ud\bv=\sum_{\rho\in W_B}\delta^{(N)}(\bv-\sqrt{t}\rho \bz_N)\ud\bv.\label{freezinglimittypeb}
\end{equation}

\end{theorem}

\begin{figure}[!t]
                \centering
                \subfloat[3 particles, $(\beta,\nu)=(2,1/2)$.]
                {\includegraphics[width=0.45\textwidth]{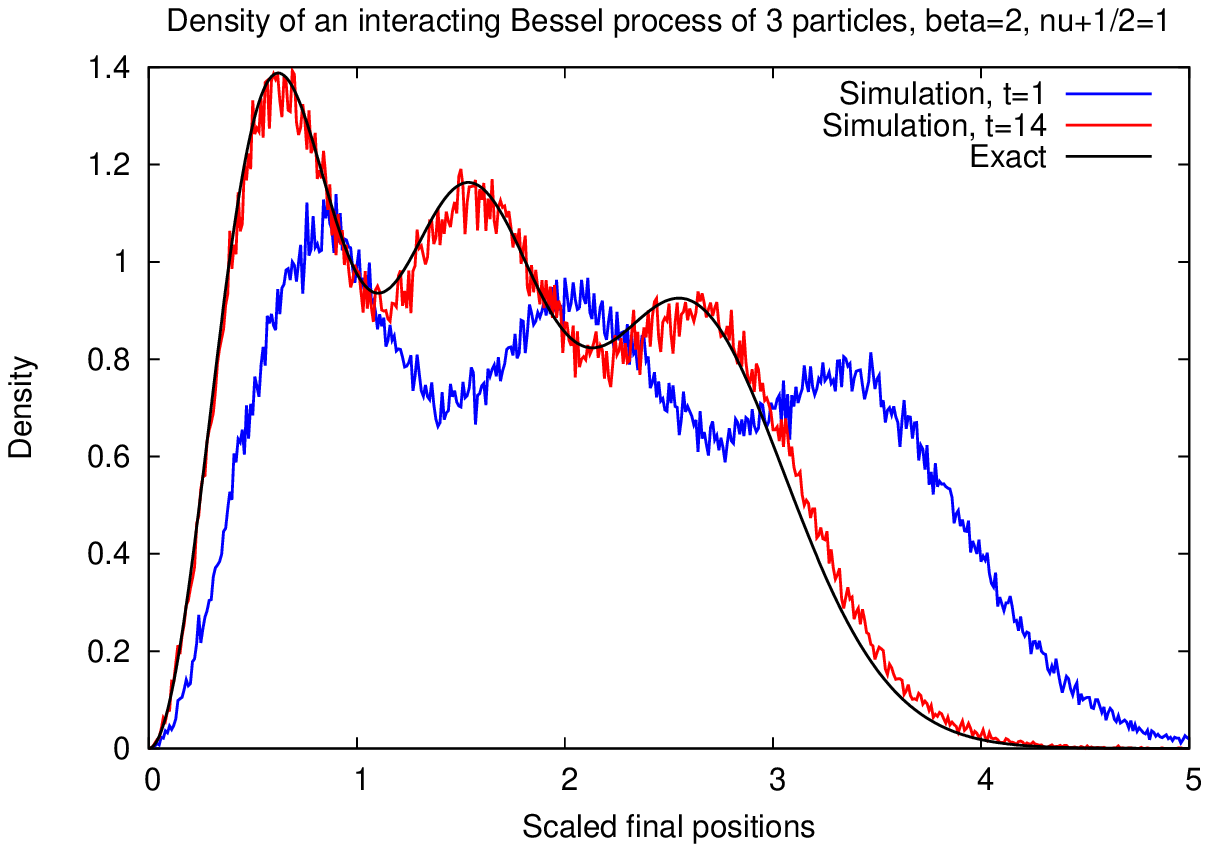}}\
                \subfloat[7 particles, $\nu=1/2$.]
                {\includegraphics[width=0.45\textwidth]{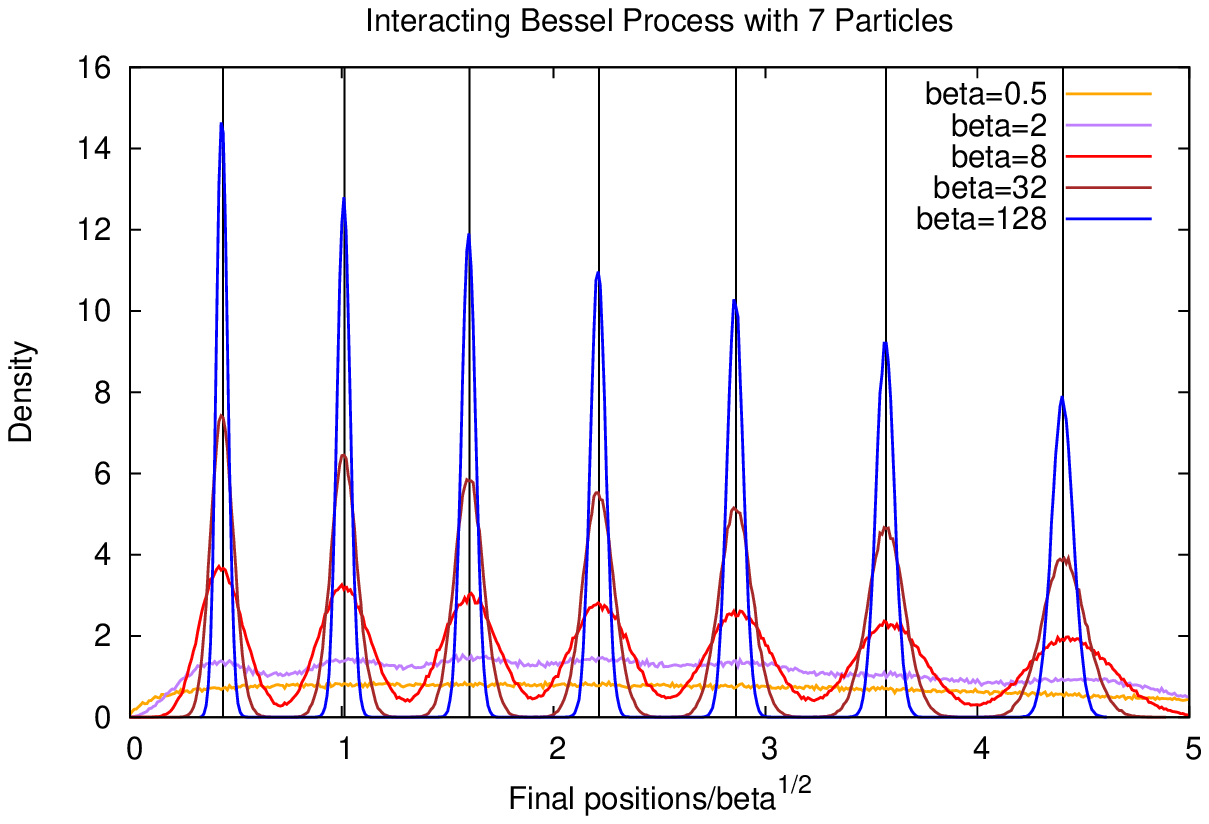}}
                \caption{(Colour online) (a) Exact and simulated densities with the initial configuration $x_i=i$, with the final positions scaled down by a factor of $\sqrt{\beta t}$. (b) Particle density of the interacting Bessel processes at $t=1$ for several values of $\beta$, and $\nu=1/2$ with initial configuration $x_i=i\times10^{-2}$.}
                \label{figRDunklB}
\end{figure}

Let us denote the positions of the particles in an interacting Bessel process by $\bY(t)=(Y_1(t),Y_2(t),\ldots,Y_N(t))$. They obey the following SDEs \cite{katoritanemura04}, 
\begin{equation}
\fl\ud Y_i(t)= \ud B_i(t) + \frac{\beta}{2}\Bigg[\frac{2\nu+1}{2Y_i(t)}+\sum_{\substack{j:j\neq i\cr j=1}}^N\Bigg\{\frac{1}{Y_i(t)-Y_j(t)}+\frac{1}{Y_i(t)+Y_j(t)}\Bigg\}\Bigg]\ud t,\label{besselSDEs}
\end{equation}
for $i=1,2,\ldots,N$. We have integrated \eref{besselSDEs} numerically to illustrate the behaviour of the processes as the value of $\beta$ increases. In Figure~\ref{figRDunklB}(a), we depict a process of three particles with the initial configuration $x_i=i$ for $\beta=2$ and $\nu=1/2$, with time increments of $2\times10^{-4}$ and $10^6$ iterations. We compare our numerical results with the exact density 
\begin{eqnarray}
K_{N,t}^{(\nu)}(y,y)dy&=&\frac{N!}{\Gamma(\nu+N)}\Bigg\{N\Bigg[L_N^{(\nu)}\Bigg(\frac{y^2}{2t}\Bigg)\Bigg]^2+L_N^{(\nu)}\Bigg(\frac{y^2}{2t}\Bigg)L_{N-1}^{(\nu)}\Bigg(\frac{y^2}{2t}\Bigg)\nonumber\\
&-&(N+1)L_{N+1}^{(\nu)}\Bigg(\frac{y^2}{2t}\Bigg)L_{N-1}^{(\nu)}\Bigg(\frac{y^2}{2t}\Bigg)\Bigg\} \frac{2}{y}\Bigg(\frac{y^2}{2t}\Bigg)^{\nu} \rme^{-y^2/2t}dy\label{exactdensity}.
\end{eqnarray}
This equation is obtained from Equation~(31) of \cite{katoritanemura04}, which states that when a non-colliding ($\beta=2$) Bessel process starts from $\bx=\bzero=(0,0\ldots 0)$, its distribution is given by the eigenvalue density of the complex Wishart ensemble of random matrices with its eigenvalues $\lambda_i$ replaced by $y^2_i/\sqrt{2t}$. Then, we obtain the one-point correlation function~\eref{exactdensity} by using Equation~(5.13) of \cite{forrester10} for Laguerre polynomials, which is the result of integrating out $N-1$ of the variables $\{y_i\}_{1\leq i\leq N}$. We observe that at $t=1$ the scaled distribution has the same overall shape of the steady-state distribution, but it is clearly different because the effect of the initial configuration is still present. For this case, the initial distribution has no tail, and so the estimate for the relaxation time given in Theorem~\ref{freezingtypeb} can be improved upon. Equation~\eref{RelaxationImprovementBeta} gives the relaxation time estimation for initial distributions with compact support; this estimation assumes that after scaling the initial distribution as $\mu(\sqrt{t}\bzeta)$, its support is completely contained within a ball of radius $\epsilon$. We denote by $r_\mu$ the point of the support of $\mu(\bx)$ which is farthest from the origin. In order to have $\epsilon<1$ with $r_\mu^2=1^2+2^2+3^2=14$, which is the case of the figure, a time larger than $t=14$ is required. In the figure, the scaled density already shows a close agreement with the exact distribution for $t=14$ and $\beta=2$. This shows that the relaxation times given here might be somewhat strict because our derivations are based on the worst-case assumption that $\mu(\bx)$ may have a long tail. Additionally, the difference between the approximated density derived from \eref{approxbeta} (not shown) and the exact density \eref{exactdensity} is never larger than $10^{-3}$, and therefore the approximated density is essentially the same curve as the exact density.

After confirming that the steady-state distribution is correct, we focus on the freezing regime. In Figure~\ref{figRDunklB}(b), we plot the particle density for 7 particles and several values of $\beta$ with $\nu=1/2$ at $t=1$ with time steps equal to $2\times10^{-4}$ and $10^6$ iterations, and the initial configuration $x_i=i\times10^{-2}$. The final positions are scaled down by $\sqrt{\beta t}$. The vertical lines denote the exact values of the square roots of the Laguerre zeroes. It is clear that as $\beta$ grows, the probability peaks become narrower and that they are centred around the vertical lines. In this case, for the relaxation time estimation $\epsilon^2>r_\mu^2/t$, we have $r_\mu^2=0.014$, and to be able to choose $\epsilon<1$, one only needs to make $t>0.014.$ Then, $t=1$ is much larger than this relaxation time estimate, so the densities in the graph for $\beta\geq 2$ correspond to the steady-state distributions. As depicted in the figure, the steady-state particle densities become a sum of delta peaks as $\beta$ tends to infinity. Note that the densities in the figure have maxima that occur in fixed points as $\beta$ varies.

The variable substitution $\lambda_i=\beta u_i^2$ maps the steady-state distribution \eref{approxbeta} into the eigenvalue distribution of the $\beta$-Laguerre ensembles of random matrices \cite{dumitriuedelman02}. The substitution yields
\begin{equation}
\fl f(t,\sqrt{\beta t}\bu)(\beta t)^{N/2}\ud\bu|_{\lambda_i=\beta u_i^2}\propto \rme^{-\sum_{i=1}^N \lambda_i/2}\prod_{i=1}^N\lambda_i^{\beta(\nu+1/2-1/\beta)/2}\prod_{1\leq i<j\leq N}|\lambda_j-\lambda_i|^{\beta}\ud\blambda. \label{betalaguerredensity}
\end{equation}
This is the eigenvalue distribution in Equation~(2) of \cite{dumitriuedelman05} with $a=\beta(N+\nu-1/2+1/\beta)/2$. Therefore, the curves in Figure~\ref{figRDunklB}(b) are consistent with the low-temperature approximations of the $\beta$-Laguerre ensembles. Note that in the distribution \eref{approxbeta}, the probability maxima are located at the square root of the zeroes of the Laguerre polynomial of parameter $\nu-1/2$, while in the case of \eref{betalaguerredensity} the probability maxima are attained when the $\{\lambda_i\}_{i=1}^N$ are equal to $\beta$ times the zeroes of $L_N^{(\nu-1/2-1/\beta)}(x)$. (See the proof of Theorem~\ref{freezingtypeb} in Section~\ref{proofs} and \ref{laguerreappendix} for details.) The difference of $1/\beta$ in the Laguerre parameter comes from the variable substitution required to obtain \eref{betalaguerredensity}, that is, the probability maxima are displaced because the Jacobian of the transformation is $\prod_{i=1}^N \lambda_i^{-1/2}$. Therefore, as the value of $\beta$ increases in \eref{approxbeta}, the probability maxima remain in the same location, while the maxima in \eref{betalaguerredensity} depend on $\beta$ and converge to the parameter $\nu-1/2$ once the freezing limit is taken. This is perfectly consistent with the limiting case (a) in Table~2 of \cite{dumitriuedelman05}.

\subsection{The limit $\nu\to\infty$ of the interacting Bessel processes}

For the regime \eref{limitgamma}, we define the following function,
\begin{equation}
\tilde{F}(\bz,\beta,N)=z^2-\beta\sum_{i=1}^N\log z_i^2+\beta N(\log \beta-1).\label{functionFtilde}
\end{equation}

\begin{theorem}\label{freezingoneta}
Assume that $\mu(\bx)$ is a Riemann-integrable distribution with power-law decay $\mu(\bx)\sim x^{-N-\eta}$ at infinity, with $\eta>0$. Let $C$ and $C^\prime$ be large positive constants. Then, for $\nu>C\beta N$, the particle distribution $f(t,\by)$ is given by
\begin{equation}
\fl f(t,\sqrt{\nu t}\bu)(\nu t)^{N/2}\ud\bu=c\,\rme^{-\nu \tilde{F}(\bu,\beta,N)/2}\nu^{N/2}\prod_{1\leq i<j\leq N}|\nu(u_j^2-u_i^2)|^\beta \ud \bu(1+O(t^{-\eta/4}))\label{approxnu}
\end{equation}
with a normalization constant $c$, provided $t> C^\prime/\beta^2\nu^2$. This distribution shows peaks in the neighbourhood of $\bu=\sqrt{\beta}(\pm 1,\ldots,\pm 1)$. Furthermore, in the regime $\nu\to\infty$, 
\begin{equation}
\lim_{\nu\to\infty}f(t,\sqrt{\nu}\bv)\nu^{N/2}\ud\bv=N!\prod_{j=1}^N\sum_{s_j=\pm1}\delta(v_j-\sqrt{\beta t}s_j)\ud\bv\label{freezinglimitoneta}
\end{equation}
for any $\mu(\bx)$ and $t>0$.
\end{theorem}

\begin{figure}[!h]
                \centering
                \subfloat[3 particles, $(\beta,\nu)=(2,16)$.]
                {\includegraphics[width=0.45\textwidth]{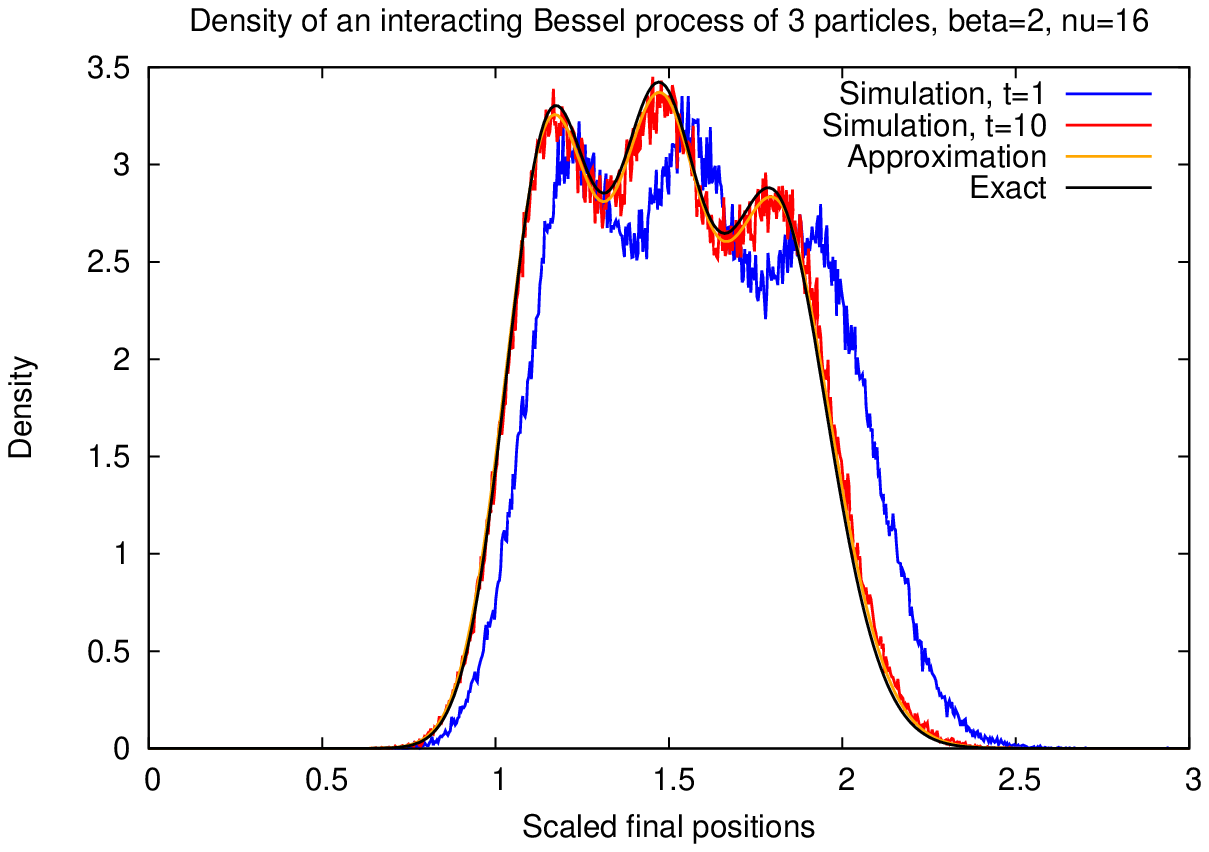}}\
                 \subfloat[7 particles, $\beta=2$]
                {\includegraphics[width=0.45\textwidth]{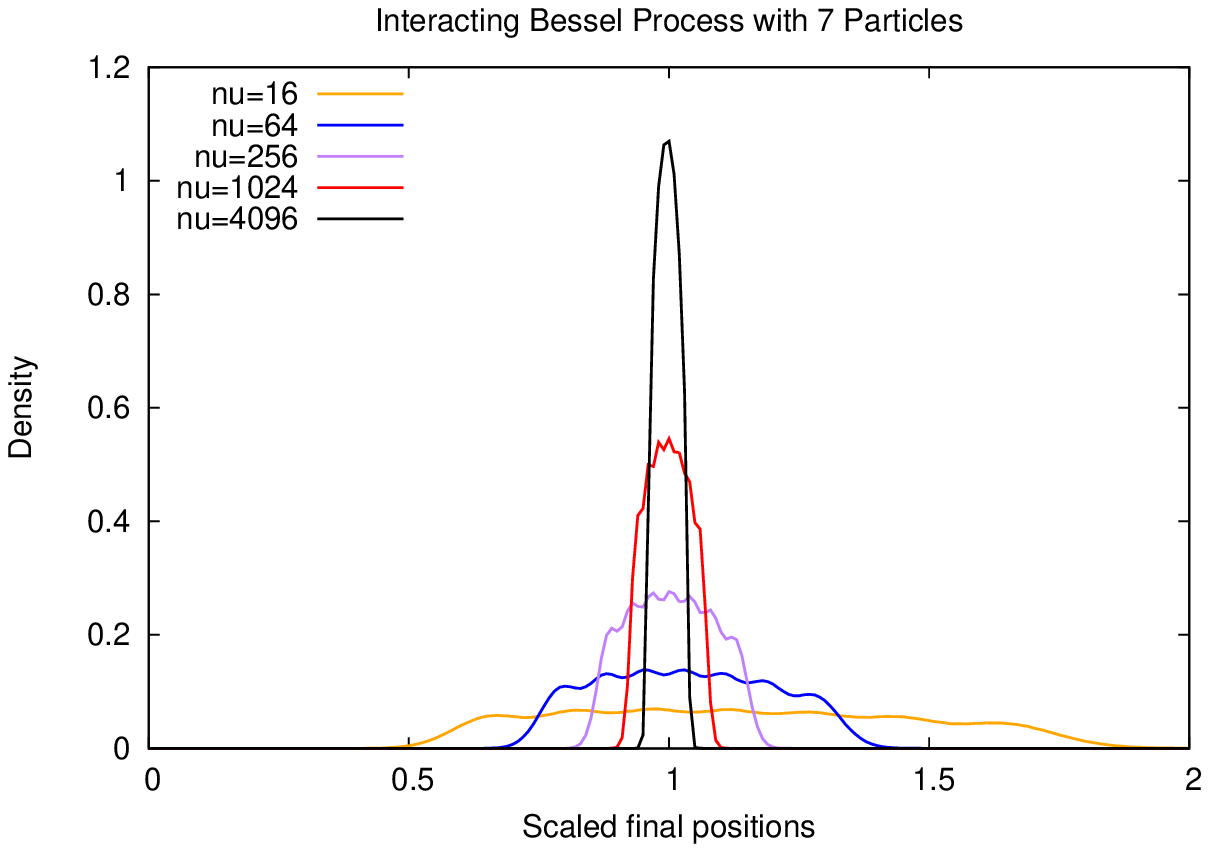}}
                \caption{(Colour online) (a) Exact, simulated and approximated densities with the initial configuration $x_i=i$ and the final positions scaled down by a factor $\sqrt{\nu t}$. (b) Particle density of the interacting Bessel process at $t=1/2$ for $\beta=2$ and several values of $\nu$ with initial configuration $x_i=i\times10^{-2}$, scaled down by a factor of $\sqrt{\beta\nu t}$.}
                \label{figRDunklBoneta}
\end{figure}

In Figure~\ref{figRDunklBoneta}(a), we plot the result of a set of numerical simulations of three particles with the initial configuration $\bx=(1,2,3)$, $\beta=2$ and $\nu=16$. We scale the final positions down by a factor of $\sqrt{\nu t}$, we use time increments of $2\times10^{-4}$ and we carry out $10^6$ iterations. The estimate for the relaxation time is again given by $1>\epsilon>r_\mu^2/t$, which gives $t>14$. The plots show that for $t=1$ the scaled density is already close to the steady-state distribution, though the effect of the initial configuration is still clear. For $t=10$ we obtain a good agreement with the exact solution given by \eref{exactdensity} in spite of the fact that $t=10$ is still smaller than our relaxation time estimate. In addition, we observe that the approximated density derived from \eref{approxnu} is also very close to the exact density \eref{exactdensity}. 

After confirming the agreement of the steady-state distribution \eref{approxnu} with the known case $\beta=2$, we investigate the behaviour of the processes in the regime \eref{limitgamma}. In Figure~\ref{figRDunklBoneta}(b), we depict the particle density as the value of $\nu$ increases for processes of seven particles with $\beta=2$ at time $t=1/2$, scaled down by a factor $\sqrt{\beta \nu t}=\sqrt{\nu}$. The initial configuration is given by $x_i=i\times10^{-2}$, and we perform $10^6$ iterations with time increments of $2\times10^{-4}$. In this case, the relaxation time is estimated to be $t>0.014$ for $\epsilon<1$, meaning that the curves in Figure~\ref{figRDunklBoneta}(b) correspond to the steady-state density for each value of $\nu$. We observe that as $\nu$ grows, the probability of finding the particles at positions close to $\sqrt{\beta t}=1$ increases, illustrating the fact that the particle density converges to a delta function centred at $\sqrt{\beta t}$ as $\nu$ tends to infinity. We observe that for $\nu=16$ and $64$, the particle densities seem slightly shifted to the right, and that for $\nu\geq 256$ the densities are centred at $\sqrt{\beta t}$. We also observe that in the regime \eref{limitgamma} the interaction between particles is negligible when compared with the repulsion from the origin. 
This occurs because, as $\nu\to\infty$, the Bessel dimension tends to infinity for all the particles, so they all move away from the origin at similar rates. In addition, the third and fourth terms on the rhs of \eref{typebradialdunklkbe} become negligible in this limit, so the process behaves approximately like a system of $N$ independent Bessel processes.

%

\section{The intertwining operator of type $B$}\label{sectionvb}

In this section, we give an expression for the intertwining operator $V_B$ derived from the generalized Bessel function of type $B$. The TPD $p(t,\by|\bx)$ is given by \cite{rosler08}
\begin{equation}
p(t,\by|\bx)=w_B\left(\frac{\by}{\sqrt{t}}\right)\frac{\rme^{-(y^2+x^2)/2t}}{c_B t^{N/2}}\sum_{\rho\in W_B}V_B \exp \left(\frac{\by\cdot\rho\bx}{t}\right).\label{TPDRadialB}
\end{equation}
Here, $w_B(\bx)$ is given by
\begin{equation}
w_B(\bx)=\prod_{i=1}^N|x_i|^{\beta(\nu+1/2)}\prod_{1\leq i<j\leq N}|x_j^2-x_i^2|^{\beta},\label{wkb}
\end{equation}
and $c_B$ is given by the Selberg integral \cite{mehta04}
\begin{eqnarray}
\fl c_B&=&\int_{\RR^N}\rme^{-x^2/2}\prod_{j=1}^N|x_j|^{\beta(\nu+1/2)}\prod_{1\leq i<j\leq N}|x_j^2-x_i^2|^{\beta}\ud\bx\nonumber\\
\fl&=&2^{N(\beta(\nu+1/2)+1)/2+\beta N(N-1)/2}\prod_{j=1}^N\frac{\Gamma(1+j\beta/2)\Gamma[\beta(\nu+j-1/2)/2+\frac{1}{2}]}{\Gamma(\beta/2+1)}.\label{selbergb}
\end{eqnarray}
The TPD \eref{TPDRadialB} is obtained by using a generalization of the Fourier transform, called the Dunkl transform, to solve the Dunkl analogue of the heat equation, given in \ref{generalreview}, equation~\eref{dunklheatsymbolic}; the solution is then made symmetric under the action of the group $W_B$. This symmetrization is the origin of the generalized Bessel function, which is the sum over $\rho$ in \eref{TPDRadialB}. The generalized Bessel function of type $B$ was first obtained in \cite{bakerforrester97}, and it is given by
\begin{equation}
 \sum_{\rho\in W_{B}}V_B \rme^{\bx\cdot\rho\by}=2^N N! \FF{2/\beta}\left(\frac{\beta}{2}(\nu+ N-1/2) +\frac{1}{2};\frac{(\bx)^2}{2},\frac{(\by)^2}{2}\right),\label{genbesseltypeb}
\end{equation}
where we use the notation $(\bx)^2=(x_1^2,\ldots,x_N^2)$. 
We give a brief explanation of the root system of type $B$, its associated operators and the reflection group $W_B$ in \ref{generalreview}. We introduce other functions and quantities necessary for our results as follows. The generalized hypergeometric function $\FF{\alpha}(b;\bx,\by)$ is given by
\begin{eqnarray}
\FF{\alpha}(b;\bx,\by)&=&\sum_{n=0}^\infty\sum_{\substack{\tau:l(\tau)\leq N\cr |\tau|=n}}\frac{c_\tau (\alpha)}{c_\tau^\prime (\alpha)}\frac{\PP{\tau}{\alpha}(\bx)\PP{\tau}{\alpha}(\by)}{(b)_\tau^{(\alpha)}(N/\alpha)_\tau^{(\alpha)}}\label{genhypergeob}.
\end{eqnarray}
Note that the dependence on $N$ of $\FF{\alpha}(b;\bx,\by)$ is implied. In this expression, $\tau$ is an integer partition of length $l(\tau)$ and modulus (or total sum) $|\tau|$ and $\PP{\tau}{\alpha}(\bx)$ is a Jack polynomial. The expression $(i,j)\in\tau$ implies that $1\leq i\leq l(\tau)$ and $1\leq j \leq \tau_i$. The quantities $c_\tau (\alpha), c_\tau^\prime (\alpha)$ and $(b)_\tau^{(\alpha)}$ are given by
\begin{eqnarray}
c_\tau(\alpha)=\prod_{(i,j)\in \tau}(\alpha(\tau_i-j)+\tau_j^\prime-i+1),\label{hooksnocorner}\\
c_\tau^\prime(\alpha)=\prod_{(i,j)\in \tau}(\alpha(\tau_i-j+1)+\tau_j^\prime-i),\label{hookscorner}
\end{eqnarray}
and
\begin{equation}\label{GPS}
(b)_\tau^{(\alpha)}=\prod_{i=1}^{l(\tau)}\frac{\Gamma(b-(i-1)/\alpha+\tau_i)}{\Gamma(b-(i-1)/\alpha)}.
\end{equation}

The monomial symmetric polynomial $m_\lambda(\bx)$ is given by
\begin{equation}
m_\lambda(\bx)=\sum_\sigma \prod_{i=1}^N x_i^{\lambda_{\sigma(i)}},
\end{equation}
where $\sigma$ is a permutation such that each of the monomials in the sum is different. Jack polynomials and monomial symmetric polynomials are related by the equation
\begin{equation}\label{JackP}
\PP{\tau}{\alpha}(\bx)=\sum_{\substack{\lambda:\lambda\leq\tau\cr |\lambda|=|\tau|}}u_{\tau\lambda}(\alpha)m_{\lambda}(\bx),
\end{equation}
where $u_{\tau\lambda}(\alpha)$ is an upper triangular matrix indexed by integer partitions with diagonal elements equal to one. The expression $\lambda\leq\tau$ refers to the natural ordering of integer partitions defined by
\begin{equation}
\lambda\leq\tau\ \Leftrightarrow\ \sum_{i=1}^j \lambda_i\leq  \sum_{i=1}^j \tau_i
\end{equation}¥
for $1\leq j\leq N$. Finally, we will use the multinomial coefficient
\begin{equation}
M(\lambda,N)=\frac{N!}{l_1^\lambda!\ldots l_p^\lambda!},
\end{equation}
which counts the number of distinct permutations of $\lambda$ assuming that it is a partition of $N$ parts (including zeroes) and $p$ distinct parts; $l_i^\lambda$ denotes the number of parts which are equal to the $i$th distinct part of $\lambda$. We refer to \cite{bakerforrester97} for details on the construction and properties of the generalized hypergeometric functions and to \cite{macdonald} (chapter 6, section 10) for details on Jack polynomials. We can now state the following.

\begin{proposition}\label{vktypeb}
The effect of $V_B$ on symmetric polynomials of the variables $\{x_i^{2}\}_{i=1}^N$ is given by
\begin{equation}
\fl V_B m_\lambda[(\bx)^2]=\frac{(2\lambda)!M(\lambda,N)}{2^{2|\lambda|}}\!\!\!\!\sum_{\substack{\tau:l(\tau)\leq N\cr |\tau|= |\lambda|}}\frac{c_\tau (2/\beta)}{c_\tau^\prime (2/\beta)}\frac{u_{\tau\lambda}(2/\beta)\PP{\tau}{2/\beta}[(\bx)^2]}{(\beta N/2)_\tau^{(2/\beta)}(\beta[\nu+N-1/2]/2+\frac{1}{2})_\tau^{(2/\beta)}}.\label{vktypebequation}
\end{equation}
\end{proposition}
We give the proof of this statement in \ref{propositionVBproof}.

\section{Proofs of Theorems}\label{proofs}

\subsection{Proof of Theorem~\ref{freezingtypeb}}

For this proof, we require two lemmas, the first concerning the intertwining operator \eref{vktypebequation} and the second concerning the generalized Bessel function \eref{genbesseltypeb}. For comparison, we list the expressions for the intertwining operator of type $A$ \cite{andrauskatorimiyashita12}:
\begin{equation}\label{VkTypeA}V_A m_\lambda(\bx)=\lambda!M(\lambda,N)\sum_{\substack{\tau:l(\tau)\leq N\cr |\tau|= |\lambda|}}\frac{c_\tau (2/\beta)}{c_\tau^\prime (2/\beta)}\frac{u_{\tau\lambda}(2/\beta)}{(\beta N/2)_\tau^{(2/\beta)}}\PP{\tau}{2/\beta}(\bx),
\end{equation}
\begin{equation}\label{frozenintertwiningoperatorA}
\lim_{\beta\to\infty}V_A m_\lambda(\bx)=\frac{M(\lambda,N)}{N^{|\lambda|}}\left(\sum_{j=1}^Nx_j\right)^{|\lambda|},
\end{equation}
\begin{equation}\label{FrozenHypergeometricFunctionA}
\lim_{\beta\to\infty}\sum_{\rho\in S_N}V_A \rme^{\by\cdot\rho\bx}=N!\exp\Big[\frac{1}{N}\Big(\sum_{i=1}^N x_i\Big)\Big(\sum_{j=1}^N y_j\Big)\Big].
\end{equation}

\begin{lemma}\label{freezingvk}
The limit \eref{limitbeta} of $V_B$ is given by
\begin{equation}\label{freezingvkequation}
\fl\lim_{\beta\to\infty}V_B\beta^{|\lambda|}m_\lambda[(\bx)^2]=\frac{(2\lambda)!M(\lambda,N)}{2^{|\lambda|}\lambda!N^{|\lambda|}(\nu+N-1/2)^{|\lambda|}}\left(\sum_{j=1}^Nx_j^2\right)^{|\lambda|}\!\!\!\!=\frac{(2\lambda)!}{\lambda!}\lim_{\beta\to\infty}V_Am_\lambda(\bu),
\end{equation}
with $\bu=(\bx)^2/[2(\nu+N-1/2)]$ on the rhs.
\end{lemma}

\begin{proof}
We multiply \eref{vktypebequation} by $\beta^{|\lambda|}=\beta^{|\tau|}$ and consider the factors inside the sum separately. First, following the proof of Theorem 3 in \cite{andrauskatorimiyashita12}, we find that the factor
\begin{equation}
\fl\frac{c_\tau(2/\beta)}{c_\tau^\prime (2/\beta)(\beta N/2)_\tau^{(2/\beta)}}=\prod_{(i,j)\in \tau}\frac{(\tau_i-j+\beta(\tau_j^\prime-i+1)/2)}{(\beta(N-i+1)/2+j-1)(\tau_i-j+1+\beta(\tau_j^\prime-i)/2)}
\end{equation}
has a freezing limit given by
\begin{equation}
\lim_{\beta\to\infty}\frac{c_\tau(2/\beta)}{c_\tau^\prime (2/\beta)(\beta N/2)_\tau^{(2/\beta)}}=\left\{
\begin{array}{rl}
0 & \text{if }l(\tau)>1,\\
\frac{1}{N^{|\tau|}|\tau|!} & \text{if }l(\tau)=1.
\end{array}\right.\label{filteringterm}
\end{equation}
In addition, we have the limit 
\begin{eqnarray}
\fl\lim_{\beta\to\infty}\frac{\beta^{|\tau|}}{(\frac{\beta}{2}[\nu+N-1/2]+\frac{1}{2})_\tau^{(2/\beta)}}&=&\lim_{\beta\to\infty}\prod_{(i,j)\in\tau}\frac{\beta}{\frac{\beta}{2}(\nu+1/2+N-i)+j-\frac{1}{2}}\nonumber\\
&=&\prod_{(i,j)\in\tau}\frac{2}{\nu+1/2+N-i},
\end{eqnarray}
but because only the term where $l(\tau)=1$ in \eref{filteringterm} survives in the limit, we only need to calculate
\begin{equation}
\lim_{\beta\to\infty}\frac{\beta^{|\tau|}}{(\beta[\nu+N-1/2]/2+\frac{1}{2})_{(|\tau|,0\ldots)}^{(2/\beta)}}=\frac{2^{|\tau|}}{(\nu+N-1/2)^{|\tau|}}.
\end{equation}
Finally, we know that (\cite{macdonald}, p. 380 combined with Equation~(43) in \cite{andrauskatorimiyashita12})
\begin{equation}
\lim_{\beta\to\infty}u_{(|\tau|,0,\ldots),\lambda}(2/\beta)\PP{(|\tau|,0,\ldots)}{2/\beta}[(\bx)^2]=\frac{|\tau|!}{\lambda!}\left(\sum_{j=1}^Nx_j^2\right)^{|\lambda|},
\end{equation}
and due to the condition in the sum in \eref{vktypebequation}, $|\tau|=|\lambda|$. Inserting all of the above limits in \eref{vktypebequation} multiplied by $\beta^{|\tau|}$ gives the first equality. For the second equality, it suffices to note that, after defining $\bu=(\bx)^2/[2(\nu+N-1/2)]$, we obtain
\begin{equation}
\frac{(2\lambda)!M(\lambda,N)}{2^{|\lambda|}\lambda!N^{|\lambda|}(\nu+N-1/2)^{|\lambda|}}\left(\sum_{j=1}^Nx_j^2\right)^{|\lambda|}=\frac{(2\lambda)!M(\lambda,N)}{\lambda!N^{|\lambda|}}\left(\sum_{j=1}^Nu_j\right)^{|\lambda|}.
\end{equation}
Comparing the rhs with \eref{frozenintertwiningoperatorA} completes the proof.
\end{proof}
Note that \eref{freezingvkequation} implies that the limit without scaling $\lim_{\beta\to\infty}V_Bm_\lambda[(\bx)^2]=0$ for any partition $\lambda$ with $|\lambda|\geq 1$. With this, we can prove the following lemma.

\begin{lemma}\label{FrozenDunklKernelB}
The $\beta\to\infty$ limit of the generalized Bessel function of type $B$ is given by
\begin{equation}
\fl\lim_{\beta\to\infty}\sum_{\rho\in W_{B}}V_B \rme^{\sqrt{\beta}\by\cdot\rho\bx}=2^NN!\exp\left(\frac{y^2 x^2}{2N(\nu+N-1/2)}\right)=2^N\lim_{\beta\to\infty}\sum_{\rho\in S_N}V_A \rme^{\bu\cdot\rho\bv},\label{vkinftyk}
\end{equation}
with $\bu=(\bx)^2/\sqrt{2(\nu+N-1/2)}$ and $\bv=(\by)^2/\sqrt{2(\nu+N-1/2)}$ on the rhs.
\end{lemma}

\begin{proof}
We begin by using \eref{exponentialexpansion} to obtain
\begin{equation}
\sum_{\rho\in W_{B}}V_B \rme^{\sqrt{\beta}\by\cdot\rho\bx}=2^NN!\sum_{\mu:l(\mu)\leq N}\frac{m_\mu[(\by)^2]}{(2\mu)!M(\mu,N)}V_B \beta^{|\mu|} m_\mu[(\bx)^2].
\end{equation}
Taking the freezing limit and using \eref{freezingvkequation} yields
\begin{eqnarray}
\fl \lim_{\beta\to\infty}\sum_{\rho\in W_{B}}V_B \rme^{\sqrt{\beta}\by\cdot\rho\bx}&=&2^NN!\sum_{\mu:l(\mu)\leq N}\frac{m_\mu[(\by)^2]}{2^{|\mu|}\mu!N^{|\mu|}(\nu+N-1/2)^{|\mu|}}\left(\sum_{j=1}^Nx_j^2\right)^{|\mu|}\nonumber\\
&=&2^NN!\sum_{\mu:l(\mu)\leq N}\frac{1}{\mu!}m_\mu\left[\frac{(\by)^2x^2}{2N(\nu+N-1/2)}\right]\nonumber\\
&=&2^NN!\exp\left(\frac{y^2x^2}{2N(\nu+N-1/2)}\right).
\end{eqnarray}
We complete the proof by substituting $\bu=(\bx)^2/\sqrt{2(\nu+N-1/2)}$ and $\bv=(\by)^2/\sqrt{2(\nu+N-1/2)}$. Comparing with \eref{FrozenHypergeometricFunctionA} gives the second equality in \eref{vkinftyk}.
\end{proof}

For the proof of Theorem~\ref{freezingtypeb}, we will need to approximate the generalized Bessel function of type $B$ for $\beta\gg 1$ finite. From Theorem~1.1 and Corollary~5.3 of \cite{rosler99}, there exists a probability measure $\mu_{\bx}(\by)$ with support contained by the convex hull of $W_B \bx=\{\rho \bx : \rho\in W_B\}$ such that 
\begin{equation}
V_B f(\bx)=\int_{\RR^N}f(\by)\ud\mu_{\bx}(\by)
\end{equation}
for any function $f(\bx)$ that is analytical and bounded at finite $\bx$. From this fact, it follows that for all values of $\bx$ and $\by$ and finite $\beta$,
\begin{equation}\label{DunklKernelBound}
|V_B \rme^{\sqrt{\beta}\bx\cdot\by}|\leq \int_{\RR^N}|\rme^{\sqrt{\beta}\bx\cdot\by}|\ud\mu_{\bx}(\by)\leq \rme^{\sqrt{\beta}xy}.
\end{equation}
This means that the result in Lemma~\ref{FrozenDunklKernelB} can be used as an approximation of the generalized Bessel function for $\beta\gg 1$ when $xy\leq2\sqrt{\beta}N(N+\nu-1/2)$. Otherwise, we can only use the bound~\eref{DunklKernelBound}.

The lower bound of $\beta$ at which the approximation of the generalized Bessel function using Lemma~\ref{FrozenDunklKernelB} is valid can be estimated by noting that
\begin{eqnarray}
\fl\Bigg\{\Delta^{(x)}+\beta\Bigg[\sum_{i=1}^N\frac{2\nu+1}{2x_i}\frac{\partial}{\partial x_i}+\sum_{1\leq i\neq j\leq N}\frac{2x_i}{x_i^2-x_j^2}\frac{\partial}{\partial x_i}\Bigg]\Bigg\}\sum_{\rho\in W_{B}}V_B \rme^{\sqrt{\beta}\by\cdot\rho\bx}&&\nonumber\\
=\sum_{\rho\in W_{B}}\beta |\rho^{-1}\by|^2 V_B \rme^{\sqrt{\beta}\by\cdot\rho\bx}=\beta y^2 \sum_{\rho\in W_{B}}V_B \rme^{\sqrt{\beta}\by\cdot\rho\bx}.&&
\end{eqnarray}
This relation follows from \eref{intertwiningrelation} and \eref{KBEkRadial} applied to the root system of type $B$ (see \ref{generalreview}). Performing the same operation on \eref{vkinftyk} gives
\begin{eqnarray}
\fl\Bigg\{\Delta^{(x)}+\beta\Bigg[\sum_{i=1}^N\frac{2\nu+1}{2x_i}\frac{\partial}{\partial x_i}+\sum_{1\leq i\neq j\leq N}\frac{2x_i}{x_i^2-x_j^2}\frac{\partial}{\partial x_i}\Bigg]\Bigg\}2^NN!\rme^{y^2 x^2/2N(\nu+N-\frac{1}{2})}&&\nonumber\\
=\Big[1+\frac{N+x^2y^2/[N(N+\nu-1/2)]}{\beta N(N+\nu-1/2)}\Big]\beta y^22^NN!\rme^{y^2 x^2/2N(\nu+N-\frac{1}{2})}.&&
\end{eqnarray}
Consequently, \eref{vkinftyk} is a reasonable approximation of the generalized Bessel function if 
\begin{equation}
\frac{N+x^2y^2/[N(N+\nu-1/2)]}{\beta N(N+\nu-1/2)}\lesssim 10^{-1}.
\end{equation}
To ensure its validity, we only use this approximation when $xy< \sqrt{\beta}N(N+\nu-1/2)/4$. When this bound is kept, the following condition may be imposed,
\begin{equation}
\frac{N+x^2y^2/[N(N+\nu-1/2)]}{\beta N(N+\nu-1/2)}\leq \frac{1}{\beta(N+\nu-1/2)}+\frac{1}{16}\lesssim 10^{-1}.
\end{equation}
Then, we require
\begin{equation}
\beta\geq\frac{C}{N+\nu-1/2}\gtrsim \frac{10}{N+\nu-1/2}\label{ConditionBetaFreezing}
\end{equation}
for some large constant $C\gtrsim 10$.

{\it Proof of Theorem~\ref{freezingtypeb}:} recall that $\alpha=\nu-1/2$. Using the expression for $p(t,\by|\bx)$ given in \eref{TPDRadialB} along with \eref{vkinftyk} and Stirling's approximation \cite{feller}, we write
\begin{eqnarray}
\fl\log \Big[\beta^{N/2}p(t,\sqrt{\beta}\bv|\bx)\Big]&\approx& \frac{\beta}{2}\Bigg[(\alpha+1)\sum_{i=1}^N\log \frac{v_i^2}{t}+2\!\!\sum_{1\leq i<j \leq N}\log\frac{|v_j^2-v_i^2|}{t}\nonumber\\
&&-\frac{v^2}{t}+N(N+\alpha)-\sum_{i=1}^N i\log i \nonumber\\
&&-\sum_{i=1}^N(\alpha+i) \log (\alpha+i)\Bigg] +\frac{N}{2}\log \frac{\beta}{2}\nonumber\\
&&-\frac{x^2}{2t}-\frac{N}{2}\log t+N\log 2+\log N!+\frac{v^2x^2}{2t^2N(N+\alpha)}.\label{logtpd}
\end{eqnarray}
Here, it has been assumed that $xv<\sqrt{\beta}t^2N(N+\alpha)/4$. The leading-order terms can be written as $-(\beta/2)F(\bv/\sqrt t,\alpha+1/2,N)$, with $F(\bz,\alpha+1/2,N)$ given by \eref{functionFforb}. The next-order term is $\frac{N}{2}\log \frac{\beta}{2}$, and the $\beta$-independent terms are written on the last line. Note that the last term corresponds to the generalized Bessel function. Because it is independent of $\beta$, it will vanish in the freezing regime. 

For the following equations, we write 
\begin{equation}
E_\beta(\bx,\by)=\frac{1}{2^N N!}\sum_{\rho\in W_B}\rme^{\sqrt{\beta}\bx\cdot\by}.
\end{equation}
Because $\beta$ is large, we use the expression
\begin{equation}
E_\beta(x,y)=\exp\Big[\frac{x^2y^2}{2N(N+\alpha)}\Big]
\end{equation}
whenever $xy<\sqrt{\beta}N(N+\alpha)/4$, and we use the bound \eref{DunklKernelBound} otherwise.

The scaled particle distribution $f(t,\sqrt{\beta } \bv)(\beta t)^{N/2}$ starting from the initial distribution $\mu(\bx)$ is given at large $\beta$ by
\begin{eqnarray}
\fl f(t,\sqrt\beta \bv)\beta^{N/2}\ud\bv&\approx&\rme^{-\beta F(\bv/\sqrt t,\nu,N)/2}\frac{\beta^{N/2}2^NN!}{(2t)^{N/2}}\int_{C_B}\rme^{-\frac{x^2}{2t}}E_\beta\Big(\frac{x}{\sqrt{t}},\frac{v}{\sqrt{t}}\Big)\mu(\bx)\ud\bx\ud\bv.\label{almosttheorem1}
\end{eqnarray}
Let us define the expectation of a test function $h(\bu)$ over $f(t,\sqrt{\beta t} \bu)(\beta t)^{N/2}$ by
\begin{equation}\label{DynamicalExpectation}
\fl\langle h\rangle_t=\int_{C_B}h(\bu)\rme^{-\beta F(\bu,\nu,N)/2}\beta^{N/2}2^{N/2}N!\int_{C_B}\rme^{-\frac{\zeta^2}{2}}E_\beta(\zeta,u)t^{N/2}\mu(\sqrt{t}\bzeta)\ud\bzeta\ud\bu,
\end{equation}
where the substitutions $\bu=\bv/\sqrt{t}$ and $\bzeta=\bx/\sqrt{t}$ have been carried out. The first objective is to show how $\langle h\rangle_t$ converges to
\begin{equation}
\langle h\rangle=\int_{C_B}h(\bu)\rme^{-\beta F(\bu,\nu,N)/2}\beta^{N/2}2^{N/2}N!\ud\bu
\end{equation}
at large $t$. Let us denote the integral over $\bzeta$ in \eref{DynamicalExpectation} by
\begin{equation}
\mathcal{I}(\bu)=\int_{C_B}\rme^{-\frac{\zeta^2}{2}}E_\beta(\zeta,u)t^{N/2}\mu(\sqrt{t}\bzeta)\ud\bzeta.
\end{equation}
To evaluate this integral, we choose a small number $0<\epsilon\ll 1$ and divide the integral into two regions, one where $\zeta<\epsilon$ (denoted $\mathcal{I}_<(\bu)$) and another where $\zeta\geq\epsilon$ (denoted $\mathcal{I}_\geq(\bu)$). The order of magnitude of the outer integral can be estimated by using the bound \eref{DunklKernelBound} as follows. Let us ``borrow'' the factor $\rme^{-\beta u^2/2}$ from $\rme^{-\beta F(\bu,\nu,N)/2}$ so we can write
\begin{equation}
|\rme^{-\beta u^2/2}\mathcal{I}_\geq(\bu)|\leq \int_\epsilon^\infty\rme^{-(\zeta-\sqrt{\beta}u)^2/2}t^{N/2}\zeta^{N-1}\tilde{\mu}(\sqrt{t}\zeta)\ud\zeta,
\end{equation}
where $\tilde{\mu}(x)$ is the result of taking the angular integral of $\mu(\bx)$ over $C_B$, and the bound \eref{DunklKernelBound} was used to complete the square in the exponential. By hypothesis on $\mu(\bx)$, there exists a positive constant $M$ such that
\begin{equation}
\mu(\bx)\leq \frac{M}{x^{N+\eta}}
\end{equation}
with $\eta>0$ for large enough $x$. From this, and assuming that $\epsilon\sqrt{t}$ is large enough to use the tail asymptotics of $\mu(\bx)$, we have
\begin{equation}
|\rme^{-\beta u^2/2}\mathcal{I}_\geq(\bu)|\leq \frac{M}{t^{\eta/2}}\int_\epsilon^\infty\rme^{-(\zeta-\sqrt{\beta}u)^2/2}\frac{\ud\zeta}{\zeta^{\eta+1}}.
\end{equation}
Using integration by parts, we see that the integral on the rhs can be approximated by
\begin{equation}
\int_\epsilon^\infty\rme^{-(\zeta-\sqrt{\beta}u)^2/2}\frac{\ud\zeta}{\zeta^{\eta+1}}=\frac{\rme^{-(\epsilon-\sqrt{\beta}u)^2/2}}{\eta \epsilon^\eta}+O(\epsilon^{-\eta+1})
\end{equation}
for $\eta\neq 1$ and a logarithmic correction if $\eta=1$. This approximation is valid for $\eta>0$, even when $\eta$ is small, provided that $\epsilon$ is chosen small enough, i.e., $\epsilon^\eta$ must be kept much smaller than 1. With this, we obtain
\begin{equation}
|\rme^{-\beta u^2/2}\mathcal{I}_\geq(\bu)|\leq \frac{M\rme^{-(\epsilon-\sqrt{\beta}u)^2/2}}{\eta t^{\eta/2} \epsilon^\eta}+O(\epsilon^{-\eta+1})\approx \frac{M\rme^{\beta u^2/2}}{\eta t^{\eta/2} \epsilon^\eta}+O(\epsilon^{-\eta+1}),
\end{equation}
because $\epsilon/\sqrt{\beta}\ll \epsilon\ll 1$, and
\begin{equation}\label{TheoremOnePieceOne}
\rme^{-\beta u^2/2}\mathcal{I}_\geq(\bu)=\rme^{-\beta u^2/2}O(t^{-\eta/2}\epsilon^{-\eta}).
\end{equation}
The inner integral $\mathcal{I}_<(\bu)$ depends on the value of $u$. Assuming that $u\geq\sqrt{\beta}N(N+\alpha)/4\epsilon$, we can borrow the factor $\rme^{-\beta u^2/2}$ again to obtain the following estimation,
\begin{equation}
\fl |\rme^{-\beta u^2/2}\mathcal{I}_<(\bu)|\leq\int_0^\epsilon \rme^{-(\zeta-\sqrt{\beta}u)^2/2}t^{N/2}\zeta^{N-1}\tilde{\mu}(\sqrt{t}\zeta)\ud\zeta\approx \rme^{-\beta u^2/2} \Big(1-\frac{M}{\eta t^{\eta/2}\epsilon^\eta}\Big),
\end{equation}
meaning that
\begin{equation}\label{TheoremOnePieceTwo}
\rme^{-\beta u^2/2}\mathcal{I}_<(\bu)=\rme^{-\beta u^2/2}[1+O(t^{-\eta/2}\epsilon^{-\eta})].
\end{equation}
On the other hand, if $u<\sqrt{\beta}N(N+\alpha)/4\epsilon$, we can use Lemma~\ref{FrozenDunklKernelB} to write
\begin{equation}
\mathcal{I}_<(\bu)=\int_0^\epsilon \rme^{-\frac{\zeta^2}{2}\left(1-\frac{u^2}{N(N+\alpha)}\right)}t^{N/2}\zeta^{N-1}\tilde{\mu}(\sqrt{t}\zeta)\ud\zeta.
\end{equation}
We borrow the factor $\rme^{-\beta u^2/2}$ one last time and we define $\gamma=N(N+\alpha)$ to obtain the following approximation,
\begin{eqnarray}
\rme^{-\beta u^2/2}\mathcal{I}_<(\bu)&=&\int_0^\epsilon \rme^{-u^2(\beta-\zeta^2/\gamma)/2}\rme^{-\zeta^2/2}t^{N/2}\zeta^{N-1}\tilde{\mu}(\sqrt{t}\zeta)\ud\zeta\nonumber\\
&=&\rme^{-u^2(\beta-O(\epsilon^2)/\gamma)/2}\rme^{-O(\epsilon^2)/2}\int_0^\epsilon t^{N/2}\zeta^{N-1}\tilde{\mu}(\sqrt{t}\zeta)\ud\zeta\nonumber\\
&\approx&\rme^{-\beta u^2/2}[1+O(t^{-\eta/2}\epsilon^{-\eta})].\label{NiceInnerPart}
\end{eqnarray}
The equality of the second line follows from the mean value theorem, and in the approximation of the last line we have assumed that $\epsilon^2\ll \beta\gamma$. Using these expressions for $\mathcal{I}(\bu)$ we can write
\begin{eqnarray}
\langle h\rangle_t&=&\int_{C_B}h(\bu)\rme^{-\beta F(\bu,\nu,N)/2}\beta^{N/2}2^{N/2}N![\mathcal{I}_<(\bu)+\mathcal{I}_\geq(\bu)]\ud\bu\nonumber\\
&=&\int_{\substack{C_B, u:\cr u<\sqrt{\beta}\gamma/4\epsilon}}h(\bu)\rme^{-\beta F(\bu,\nu,N)/2}\beta^{N/2}2^{N/2}N!\ud\bu[1+O(t^{-\eta/2}\epsilon^{-\eta})]\nonumber\\
&&+\int_{\substack{C_B, u:\cr u\geq\sqrt{\beta}\gamma/4\epsilon}}h(\bu)\rme^{-\beta F(\bu,\nu,N)/2}\beta^{N/2}2^{N/2}N!\ud\bu[1+O(t^{-\eta/2}\epsilon^{-\eta})]\nonumber\\
&&+\int_{C_B}h(\bu)\rme^{-\beta F(\bu,\nu,N)/2}\beta^{N/2}2^{N/2}N!\mathcal{I}_\geq(\bu)\ud\bu\nonumber\\
&=&\langle h\rangle[1+O(t^{-\eta/2}\epsilon^{-\eta})].\label{DynamicalExpectationConverging}
\end{eqnarray}
The proof of the first statement is completed by choosing a relationship between $t$ and $\epsilon$. Because $\epsilon$ must be chosen so that $\epsilon\ll 1$ and $\epsilon\sqrt{t}\gg 1$ for long times $t$, we set $\epsilon\propto t^{-1/4}$ so that with growing $t$, $\epsilon$ becomes arbitrarily small while $\sqrt{t}\epsilon\propto t^{1/4}$ becomes arbitrarily large. This produces the correction $O(t^{-\eta/4})$ in Equation~\eref{approxbeta}. Note also that we have assumed that $\epsilon\ll\sqrt{\beta N(N+\alpha)}$, which gives the requirement $t> C^\prime/[\beta^2 N^2(N+\alpha)^2]$, where $C^\prime$ is a large positive constant (of minimum order $10^1$). Due to the use of Lemma~\ref{FrozenDunklKernelB}, the lower bound for $\beta$ is given by \eref{ConditionBetaFreezing}.

A great improvement for the error terms can be obtained if the initial distribution $\mu(\bx)$ has compact support. In that case, there is no tail and it is enough to set a small fixed value of $\epsilon$ and take a time $t$ large enough that $\sqrt{t}\epsilon$ is larger than the distance from the farthest point in the support of $\mu(\bx)$ to the origin. If we define the support of $\mu(\bx)$ by $\mathcal{S}_\mu$, then we define the distance to its farthest point by
\begin{equation}\label{MaximumRadius}
r_\mu=\sup_{\bx\in\mathcal{S}_\mu}|\bx|.
\end{equation}
Then, the relaxation time for a given $\epsilon$ is
\begin{equation}\label{RelaxationImprovementBeta}
t\gtrsim r_\mu^2/\epsilon^2,
\end{equation}
with a correction of order $O(\epsilon^2)/\beta\gamma$ to the main Gaussian, as given by the second line of \eref{NiceInnerPart}.

For the second statement, let us recall the expression for $f(t,\sqrt{\beta t} \bu)(\beta t)^{N/2}$,
\begin{equation}\label{DynamicalDistribution}
\fl f(t,\sqrt{\beta t} \bu)(\beta t)^{N/2}=\rme^{-\beta F(\bu,\nu,N)/2}\beta^{N/2}2^{N/2}N!\int_{C_B}\rme^{-\frac{\zeta^2}{2}}E_\beta(\zeta,u)t^{N/2}\mu(\sqrt{t}\bzeta)\ud\bzeta.
\end{equation}
We calculate the freezing limit of the expression outside the integral as follows. As shown in \ref{laguerreappendix}, the function $F(\bz,\nu,N)$ is convex for $\bz\in\RR^N\backslash\{\bzero\}$ (from Equation~\eref{HessianFPositiveDefinite}, it follows that its Hessian is non-negative). Its minima occur at the solutions of the set of $N$ equations
\begin{equation}\label{minconditionb}
z_i^2=\alpha+1+\sum_{\substack{j:j\neq i\cr j=1}}^N\frac{2z_i^2}{z_i^2-z_j^2},\quad i=1,\ldots,N.
\end{equation}
These equations do not change if the variables $\{z_i\}_{i=1}^N$ are permuted or if their sign is changed. This system of $N$ equations is solved by the roots $\{s_{i,\alpha}\}_{1\leq i\leq N}$ of the $N$th associated Laguerre polynomial of parameter $\alpha=\nu-1/2$ by setting $z_i^2=s_{i,\alpha}$ (see equations \eref{minimizationonr} to \eref{LaguerreDifferentialEquation}). In fact, due to \eref{partone}, \eref{parttwo} and \eref{partthree}, the minimum value of $F(\bz,\alpha+1/2,N)$ is zero. It follows that 
\begin{equation}
\lim_{\beta\to\infty}\exp\Bigg[-\frac{\beta}{2}F(\bu,\nu,N)\Bigg]\Bigg(\frac{\beta}{2}\Bigg)^{N/2}\propto \sum_{\rho\in W_B}\delta^{(N)}(\bu-\rho \bz_N)\label{Fapprox}
\end{equation}
with $\bu=\bv/\sqrt{t}$. Using \eref{Fapprox} to calculate the freezing limit of \eref{DynamicalDistribution}, where we can use Lemma~\ref{FrozenDunklKernelB} for all values of $\bu$, we obtain
\begin{eqnarray}
\fl\lim_{\beta\to\infty}f(t,\sqrt{\beta t} \bu)(\beta t)^{N/2}\ud\bu&\propto&\sum_{\rho\in W_B}\delta^{(N)}(\bu-\rho \bz_N)\int_{C_B}\rme^{-\frac{x^2}{2t}\left(1-\frac{u^2}{\gamma}\right)}\mu(\bx)\ud\bx\ud\bu.
\end{eqnarray}
We only need the value of the integral over $\bx$ when $\bu=\rho\bz_N$. In that case, $u^2=N(N+\alpha)=\gamma$ (see \eref{partone}), and the integral is equal to one. We find that equality holds because $f(t,\sqrt{\beta t} \bu)(\beta t)^{N/2}\ud\bu$ is normalized to $|W_B|=N!2^N$ when integrated over $\RR^N$. Setting $\bu=\bv/\sqrt{t}$ proves the second part of the statement.\qquad\qquad\qquad\qquad\ \ $\square$

\subsection{Proof of Theorem~\ref{freezingoneta}}

As in the proof of Theorem~\ref{freezingtypeb}, we present the proof of Theorem~\ref{freezingoneta} immediately after proving Lemmas~\ref{freezingvketa} and \ref{genbesselbgamma}.

\begin{lemma}\label{freezingvketa}
In the limit $\nu\to\infty$, we have
\begin{equation}\label{anotheretainfinity}
\lim_{\nu\to\infty}V_B \nu^{|\lambda|}m_\lambda[(\bx)^2]=\frac{(2\lambda)!}{\lambda!}V_A m_\lambda(\bu),
\end{equation}
with $\bu=(\bx)^2/(2\beta)$.
\end{lemma}
\begin{proof}
We use \eref{vktypebequation} and notice that we only need to calculate the following limit ($|\lambda|=|\tau|$),
\begin{equation}\label{NuOneTermLimit}
\fl\lim_{\nu\to\infty}\!\frac{(\frac{\beta}{2}[\nu+N-1/2]+\frac{1}{2})_\tau^{(2/\beta)}}{\nu^{|\tau|}}=\!\lim_{\nu\to\infty}\prod_{(i,j)\in\tau}\frac{\frac{\beta}{2}(\nu+1/2+N-i)+j-\frac{1}{2}}{\nu}=\Bigg(\frac{\beta}{2}\Bigg)^{|\tau|}.
\end{equation}
This yields
\begin{equation}
\fl\lim_{\nu\to\infty}V_B \nu^{|\lambda|}m_\lambda[(\bx)^2]=(2\lambda)!M(\lambda,N)\sum_{\substack{\tau:l(\tau)\leq N\cr |\tau|= |\lambda|}}\frac{c_\tau (2/\beta)}{c_\tau^\prime (2/\beta)}\frac{u_{\tau\lambda}(2/\beta)}{(\beta N/2)_\tau^{(2/\beta)}}\frac{\PP{\tau}{2/\beta}[(\bx)^2]}{(2\beta)^{|\tau|}},
\end{equation}
and comparing this expression with \eref{VkTypeA}, the proof is complete.
\end{proof}

For the following lemma, we use the generalized hypergeometric function
\begin{equation}\label{ZeroFZero}
\FZ{\alpha}(\bx,\by)=\sum_{n=0}^\infty\sum_{\substack{\tau:l(\tau)\leq N\cr |\tau|=n}}\frac{c_\tau (\alpha)}{c_\tau^\prime (\alpha)}\frac{\PP{\tau}{\alpha}(\bx)\PP{\tau}{\alpha}(\by)}{(N/\alpha)_\tau^{(\alpha)}}.\label{zeroFzero}
\end{equation}

\begin{lemma}\label{DunklKernelLimitNu}
The limit \eref{limitgamma} of the generalized Bessel function of type $B$ is given by
\begin{equation}
\lim_{\nu\to\infty}\sum_{\rho\in W_{B}}V_B \rme^{\sqrt{\nu}\by\cdot\rho\bx}=2^NN!\FZ{2/\beta}\Bigg(\frac{(\bx)^2}{\sqrt{2\beta}},\frac{(\by)^2}{\sqrt{2\beta}}\Bigg)=2^N\sum_{\rho\in S_N}V_A \rme^{\bu\cdot\rho\bv},\label{vkbtransitiontoa}
\end{equation}
with $\bu=(\bx)^2/\sqrt{2\beta}$ and $\bv=(\by)^2/\sqrt{2\beta}$ on the rhs.\label{genbesselbgamma}
\end{lemma}

\begin{proof}
This is calculated directly from \eref{exponentialexpansion} and \eref{anotheretainfinity}. We have
\begin{eqnarray}
\fl\lim_{\nu\to\infty}\sum_{\rho\in W_{B}}V_B \rme^{\sqrt{\nu}\by\cdot\rho\bx}&=&\sum_{\substack{\mu:l(\mu)\leq N}}\frac{2^NN!}{(2\mu)!}\frac{m_{\mu}[(\by)^2]}{M(\mu,N)}\lim_{\nu\to\infty}V_B\nu^{|\mu|}m_{\mu}[(\bx)^2]\nonumber\\
&=&\sum_{\substack{\mu:l(\mu)\leq N}}\!\!\!\!\frac{2^NN!m_{\mu}[(\by)^2]}{2^{2|\mu|}}\!\!\!\!\sum_{\substack{\tau:l(\tau)\leq N\cr |\tau|= |\mu|}}\frac{c_\tau (2/\beta)}{c_\tau^\prime (2/\beta)}\frac{u_{\tau\mu}(2/\beta)}{(\beta N/2)_\tau^{(2/\beta)}}\frac{\PP{\tau}{2/\beta}[(\bx)^2]}{(\beta/2)^{|\tau|}}\nonumber\\
&=&2^NN!\sum_{\substack{\tau:l(\tau)\leq N}}\frac{c_\tau (2/\beta)}{c_\tau^\prime (2/\beta)}\frac{\PP{\tau}{2/\beta}[(\by)^2]\PP{\tau}{2/\beta}[(\bx)^2]}{2^{|\tau|}\beta^{|\tau|}(\beta N/2)_\tau^{(2/\beta)}},
\end{eqnarray}
and we absorb the factors $2^{|\tau|}$ and $\beta^{|\tau|}$ in the denominator into the Jack polynomials using their homogeneity, obtaining \eref{vkbtransitiontoa}.
\end{proof}

Note that in this case, the bound \eref{DunklKernelBound} becomes
\begin{equation}\label{DunklKernelBoundNu}
|V_B \rme^{\sqrt{\nu}\bx\cdot\by}|\leq \rme^{\sqrt{\nu}xy}.
\end{equation}
The result in Lemma~\ref{DunklKernelLimitNu} can be used as an approximation of the generalized Bessel function whenever $\nu$ is large enough to validate the approximation of \eref{NuOneTermLimit} given by\begin{equation}\label{ReasonForBoundOnNu}
\prod_{(i,j)\in\tau}\Bigg[\frac{\beta}{2}+\frac{\beta}{2\nu}\Bigg(N+\frac{1}{2}-i\Bigg)+\frac{1}{2\nu}+\frac{j-1}{\nu}\Bigg]\approx\Big(\frac{\beta}{2}\Big)^{|\tau|}.
\end{equation}
This is valid if $\beta N/\nu$ is small (say, smaller than $10^{-1}$), while assuming that the integer partition $\tau$ is a partition of an integer of order $N$. To guarantee that $|\tau|$ is of order $N$, we restrict $\bx$ and $\by$ to values for which \eref{DunklKernelBoundNu} is valid and far from equality. This restriction is discussed after \eref{StartDiscussionBounds} in the proof of Theorem~\ref{freezingoneta}.

{\it Proof of Theorem~\ref{freezingoneta}:} As in the proof of Theorem~\ref{freezingtypeb}, we consider \eref{TPDRadialB} for large values of $\nu$ and we write
\begin{eqnarray}
\fl p(t,\sqrt{\nu}\bv|\bx)\nu^{N/2}&\propto& \exp\Bigg\{-\frac{\nu}{2}\Bigg[\frac{v^2}{t}-\beta\sum_{i=1}^N\log \frac{v_i^2}{t}+\beta N(\log \beta-1)\Bigg]\Bigg\}\frac{\nu^{N/2}}{t^{N/2}}\nonumber\\
&\times&\prod_{1\leq i<j\leq N}\Big|\nu\frac{v_j^2-v_i^2}{t}\Big|^\beta\FZ{2/\beta}\Bigg(\frac{(\bx)^2}{t\sqrt{2\beta}},\frac{(\bv)^2}{t\sqrt{2\beta}}\Bigg)\rme^{-x^2/2t}.
\end{eqnarray}
The quantity in square parentheses is the function $\tilde{F}(\bu=\bv/\sqrt{t},\beta,N)$, given in \eref{functionFtilde}. Note that the term $\beta N(\log \beta -1)$ arises naturally from the dominant terms in $\log c_B$ when $\nu\gg N$ and that we have assumed that $\bx$ and $\bv$ are small enough to use Lemma~\ref{DunklKernelLimitNu}. 

We first focus on the relaxation to the steady state. Setting $E_\nu(\bx,\by)=(2^N N!)^{-1}\sum_{\rho\in W_B}\exp(\sqrt{\nu}\by\cdot\rho\bx)$, the scaled distribution in this case is given by
\begin{eqnarray}
f(t,\sqrt{\nu t}\bu)(\nu t)^{N/2}&\approx&\rme^{-\nu\tilde{F}(\bu,\beta,N)/2}\prod_{1\leq i<j\leq N}|\nu(u_j^2-u_i^2)|^\beta(2\nu)^{N/2}N!\nonumber\\
&&\times\int_{C_B}\rme^{-x^2/2t}E_\nu(\bx/\sqrt{t},\bu)\mu(\bx)\ud\bx.
\end{eqnarray}
Like in the proof of Theorem~\ref{freezingtypeb}, we consider the expectation
\begin{eqnarray}
\langle h\rangle_t^\prime&=&\int_{C_B}h(\bu)\rme^{-\nu\tilde{F}(\bu,\beta,N)/2}\prod_{1\leq i<j\leq N}|\nu(u_j^2-u_i^2)|^\beta(2\nu)^{N/2}N!\nonumber\\
&&\times\int_{C_B}\rme^{-\zeta^2/2}E_\nu(\bzeta,\bu)t^{N/2}\mu(\sqrt{t}\bzeta)\ud\bzeta\ud\bu,
\end{eqnarray}
and focus on its convergence to the steady-state expectation
\begin{equation}
\langle h\rangle^\prime=\int_{C_B}h(\bu)\rme^{-\nu\tilde{F}(\bu,\beta,N)/2}\prod_{1\leq i<j\leq N}|\nu(u_j^2-u_i^2)|^\beta(2\nu)^{N/2}N!\ud\bu.
\end{equation}
Following the same same procedure of the proof of Theorem~\ref{freezingtypeb}, we denote by $\mathcal{I}(\bu)$ the integral over $\bzeta$ and define a small number $0<\epsilon\ll 1$ to separate it into the regions $\zeta<\epsilon$ and  $\zeta\geq\epsilon$. The analogue of \eref{TheoremOnePieceOne} is obtained from the bound \eref{DunklKernelBoundNu}
\begin{equation}
\rme^{-\nu u^2/2}\mathcal{I}_\geq(\bu)=\rme^{-\nu u^2/2}O(t^{-\eta/2}\epsilon^{-\eta}),
\end{equation}
and the inner part of $\mathcal{I}(\bu)$, much like \eref{TheoremOnePieceTwo}, is approximated by 
\begin{equation}
\rme^{-\nu u^2/2}\mathcal{I}_<(\bu)=\rme^{-\nu u^2/2}[1+O(t^{-\eta/2}\epsilon^{-\eta})]\label{StartDiscussionBounds}
\end{equation}
for $u\geq\sqrt{\nu}\beta/4\epsilon$. The latter bound for $u$ is an estimation of the values at which we can use Lemma~\ref{DunklKernelLimitNu} as an approximation of $E_\nu(\bzeta,\bu)$ while keeping the bound \eref{DunklKernelBoundNu}. The reason for this is as follows. $\FZ{2/\beta}((\bx)^2/\sqrt{2\beta},(\by)^2/\sqrt{2\beta})$ is a positive and monotonic function for growing $x$ or $y$, by construction. Furthermore, the $n$th term in the expansion of $\FZ{2/\beta}((\bx)^2/\sqrt{2\beta},(\by)^2/\sqrt{2\beta})$, see \eref{ZeroFZero}, is of order $x^{2n}y^{2n}$, while the $n$th term in the expansion of the rhs of \eref{DunklKernelBoundNu} is of order $\sqrt{\nu}x^ny^n$. This means that one can always find $\bx$ and $\by$ large enough in magnitude so that $\FZ{2/\beta}((\bx)^2/\sqrt{2\beta},(\by)^2/\sqrt{2\beta})\geq\rme^{\sqrt{\nu} x y}$. The problem is that the minimum magnitude of $\bx$ and $\by$ required for this inequality to hold depends on the directions of the two vectors. There is one case, which follows from Equations~(2.8) and (3.2b) of \cite{bakerforrester97}, where it is easy to extract a condition for $\bx$ and $\by$, or in this case, $\bzeta$ and $\bu$,
\begin{equation}\label{ZeroFZeroSpecialProperty}
\FZ{2/\beta}\Bigg(\frac{(\bzeta)^2}{2},\frac{u^2\bone}{\beta}\Bigg)=\FZ{2/\beta}\Bigg(\frac{u^2(\bzeta)^2}{2\beta},\bone\Bigg)=\exp\Bigg(\frac{\zeta^2u^2}{2\beta}\Bigg)\geq\rme^{\sqrt{\nu} x y}.
\end{equation}
This implies $u\zeta\geq2\sqrt{\nu}\beta$. To ensure that the required bounds are satisfied, we will approximate $E_\nu(\bzeta,\bu)$ using Lemma~\ref{DunklKernelLimitNu} whenever $u\zeta<\sqrt{\nu}\beta/4$ and we will assume that the error involved in this approximation decreases with decreasing values of $u\zeta$.

If $u<\sqrt{\nu}\beta/4\epsilon$, then we obtain the following estimate for the inner part of $\mathcal{I}(\bu)$,
\begin{eqnarray}
\fl\rme^{-\nu u^2/2}\mathcal{I}_<(\bu)&=&\int_{\substack{C_B:\cr \zeta<\epsilon}} \rme^{-\nu u^2/2}\FZ{2/\beta}\Bigg(\frac{(\bzeta)^2}{2},\frac{(\bu)^2}{\beta}\Bigg)\rme^{-\zeta^2/2}\mu(\sqrt{t}\bzeta)\ud\bzeta\nonumber\\
\fl&=&\rme^{-\nu u^2/2}\FZ{2/\beta}\Bigg(\frac{O(\epsilon^2)}{2},\frac{(\bu)^2}{\beta}\Bigg)\rme^{-O(\epsilon^2)/2}\int_0^\epsilon t^{N/2}\zeta^{N-1}\tilde{\mu}(\sqrt{t}\zeta)\ud\zeta\nonumber\\
\fl&\approx&\rme^{-\nu u^2/2}[1+O(t^{-\eta/2}\epsilon^{-\eta})].
\end{eqnarray}
The last line is justified by imposing the condition that $\epsilon^2\ll \beta\nu$, so the Gaussian dominates over the hypergeometric function. Finally, the analogue of Equation~\eref{DynamicalExpectationConverging} becomes
\begin{equation}
\langle h\rangle_t^\prime=\langle h\rangle^\prime[1+O(t^{-\eta/2}\epsilon^{-\eta})].
\end{equation}
Here, we can set $\epsilon\propto t^{-1/4}$, which gives the restriction $t>C^\prime/\beta^2\nu^2$ with $C^\prime\gtrsim 10^1$. This is valid only when $\beta N/\nu$ is sufficiently small, as discussed after \eref{ReasonForBoundOnNu}, meaning that $\nu>C\beta N$, again with $C\gtrsim 10^1$. This completes the proof of the first statement. For distributions with compact support, the same argument from the proof of Theorem~\ref{freezingtypeb} gives an estimation of the relaxation time, which is $t\gtrsim r_\mu^2/\epsilon^2$ with $r_\mu$ given by Equation~\eref{MaximumRadius} and a correction to the main Gaussian of order $O(\epsilon^2)/\beta\nu$.

We now focus on the limit $\nu\to\infty$. We observe that $\tilde{F}(\bu,\beta,N)$ is a convex function,
\begin{equation}
\frac{\partial^2 \tilde{F}}{\partial u_j \partial u_i}=2\delta_{ij}\left[1+\frac{\beta}{u_i^2}\right].
\end{equation}
Its minima lie at
\begin{equation}
u_i=\pm\sqrt{\beta},\ i=1,\ldots,N,
\end{equation}
and its minimum value is zero. Then, for $\nu\gg N$ the following approximation holds,
\begin{equation}
\rme^{-\nu\tilde{F}(\bu,\beta,N)/2}\approx\prod_{i=1}^N\sum_{s_i=\pm1}\exp[- \nu(u_i-s_i\sqrt{\beta})^2].
\end{equation}
Consider now the integral 
\begin{eqnarray}
\mathcal{E}&=&\int_{\bar{C_B}}h(\bu)\prod_{i=1}^N\sum_{s_i=\pm1}\rme^{- \nu(u_i-s_i\sqrt{\beta})^2}\prod_{1\leq i<j\leq N}|\nu(u_j^2-u_i^2)|^\beta(2\nu)^{N/2}N!\nonumber\\
&&\times\int_{\bar{C_B}}\rme^{-x^2/2t}\FZ{2/\beta}\Bigg(\frac{(\bx)^2}{2t},\frac{(\bu)^2}{\beta}\Bigg)\mu(\bx)\ud\bx\ud\bu,
\end{eqnarray}
where $h(\bu)$ is a test function with polynomial growth at infinity and $\bar{C_B}$ is the closure of $C_B$. Define the following subset of $\bar{C_B}$, $\mathcal{D}_\epsilon=\{\by\in C_B: \sqrt{\beta}-\epsilon\leq y_1\leq\ldots\leq y_N\leq \sqrt{\beta}+\epsilon\}$, where $0<\epsilon\ll 1$. We use Lemma~\ref{DunklKernelLimitNu} because we will be taking the limit $\nu\to\infty$ shortly. At very large values of $\nu$, one has
\begin{eqnarray}
\int_{\bar{C_B}\backslash\mathcal{D}_\epsilon}h(\bu)\prod_{i=1}^N\sum_{s_i=\pm1}\rme^{- \nu(u_i-s_i\sqrt{\beta})^2}\prod_{1\leq i<j\leq N}|\nu(u_j^2-u_i^2)|^\beta(2\nu)^{N/2}N!\nonumber\\
\quad\times\int_{\bar{C_B}}\rme^{-x^2/2t}\FZ{2/\beta}\Bigg(\frac{(\bx)^2}{2t},\frac{(\bu)^2}{\beta}\Bigg)\mu(\bx)\ud\bx\ud\bu\nonumber\\
\quad\qquad=O[\rme^{-\nu\epsilon^2}],
\end{eqnarray}
because the Gaussian term dominates the integrand away from $(\bu)^2=\sqrt{\beta}\bone=\sqrt{\beta}(1,\ldots,1)$. Therefore, if $\epsilon$ is chosen small while keeping $\nu\epsilon^2$ very large, this part of the integral can be neglected. For this purpose, set $\epsilon\propto \nu^{-\omega}$ with $0<\omega<1/2$. Then, the integral over $\mathcal{D}_\epsilon$ is simplified using the mean value theorem,
\begin{eqnarray}
\int_{\mathcal{D}_\epsilon}h(\bu)\prod_{i=1}^N\sum_{s_i=\pm1}\rme^{- \nu(u_i-s_i\sqrt{\beta})^2}\prod_{1\leq i<j\leq N}|\nu(u_j^2-u_i^2)|^\beta(2\nu)^{N/2}N!\nonumber\\
\quad\quad\times\int_{\bar{C_B}}\rme^{-x^2/2t}\FZ{2/\beta}\Bigg(\frac{(\bx)^2}{2t},\frac{(\bu)^2}{\beta}\Bigg)\mu(\bx)\ud\bx\ud\bu\nonumber\\
\quad=(2\nu)^{N/2}N!\prod_{1\leq i<j\leq N}|\nu(u_{j*}^2-u_{i*}^2)|^\beta\int_{\mathcal{D}_\epsilon}h(\bu)\prod_{i=1}^N\sum_{s_i=\pm1}\rme^{- \nu(u_i-s_i\sqrt{\beta})^2}\nonumber\\
\quad\quad\times\int_{\bar{C_B}}\rme^{-x^2/2t}\FZ{2/\beta}\Bigg(\frac{(\bx)^2}{2t},\frac{(\bu)^2}{\beta}\Bigg)\mu(\bx)\ud\bx\ud\bu,
\end{eqnarray}
where $\bu_*\in\mathcal{D}_\epsilon$. Then, the components of $\bu_*$ have the property that
\begin{equation}
u_{i*}=\sqrt{\beta}+O(\epsilon),
\end{equation}
and consequently
\begin{equation}
u_{j*}^2-u_{i*}^2=\sqrt{\beta}O(\epsilon)+O(\epsilon^2)=O(\epsilon).
\end{equation}
Thus, the order of magnitude of the product of differences is given by
\begin{equation}
\prod_{1\leq i<j\leq N}|\nu(u_{j*}^2-u_{i*}^2)|^\beta=\prod_{1\leq i<j\leq N}|O(\nu\epsilon)|^\beta=O(\nu^{(1-\omega)\beta N(N-1)/2}).
\end{equation}
This means that as $\nu\to\infty$, the product of differences tends to infinity (as opposed to vanishing, as all variables $\{u_{i*}\}_{1\leq i\leq N}$ tend to a single value). Therefore, it makes sense to write 
\begin{eqnarray}
\fl\lim_{\nu\to\infty}\mathcal{E}\propto\int_{\bar{C_B}}h(\bu)\prod_{i=1}^N\sum_{s_i=\pm1}\delta(u_i-s_i\sqrt{\beta})\int_{\bar{C_B}}\rme^{-x^2/2t}\FZ{2/\beta}\Bigg(\frac{(\bx)^2}{2t},\frac{(\bu)^2}{\beta}\Bigg)\mu(\bx)\ud\bx\ud\bu\nonumber\\
\fl\qquad=\int_{\bar{C_B}}h(\bu)\prod_{i=1}^N\sum_{s_i=\pm1}\delta(u_i-s_i\sqrt{\beta})\ud\bu\int_{\bar{C_B}}\rme^{-x^2/2t}\FZ{2/\beta}\Bigg(\frac{(\bx)^2}{2t},\bone\Bigg)\mu(\bx)\ud\bx\nonumber\\
=h(\sqrt{\beta}\bone)\int_{\bar{C_B}}\rme^{-x^2/2t}\FZ{2/\beta}\Bigg(\frac{(\bx)^2}{2t},\bone\Bigg)\mu(\bx)\ud\bx.
\end{eqnarray}
Making use of Equation~\eref{ZeroFZeroSpecialProperty} with $\bzeta=\bx/\sqrt{t}$ and $u^2\bone/\beta$ replaced by $\bone$ finally gives
\begin{equation}
\lim_{\nu\to\infty}\mathcal{E}\propto h(\sqrt{\beta}\bone)\int_{\bar{C_B}}\rme^{-x^2/2t}\rme^{x^2/2t}\mu(\bx)\ud\bx=h(\sqrt{\beta}\bone),
\end{equation}
or, in the sense of distributions,
\begin{equation}
\lim_{\nu\to\infty}f(t,\sqrt{\nu t}\bu)(\nu t)^{N/2}\ud\bu\propto\delta^{(N)}(\bu-\sqrt{\beta}\bone)\ud\bu.
\end{equation}
Because only the closure of $\bar{C_B}$ was considered, the only delta function that survived in the proof was the one located at $\sqrt{\beta}\bone$. However, if we extend this analysis to $\RR^N$, we see that
\begin{equation}
\lim_{\nu\to\infty}f(t,\sqrt{\nu t}\bu)(\nu t)^{N/2}\ud\bu=N!\prod_{i=1}^N\sum_{s_i=\pm1}\delta(u_i-s_i\sqrt{\beta})\ud\bu,
\end{equation}
where the proportionality constant is $N!$ because both members of the expression are normalized to $2^N N!$ in $\RR^N$.\qquad\qquad\qquad\qquad\qquad\qquad\qquad\qquad\qquad\qquad\qquad\qquad$\square$

\section{Concluding Remarks}\label{conclusions}

We obtained the time-scaled steady-state distribution of the interacting Bessel processes for the cases where $\beta$ or $\nu$ are large but finite. In particular, we showed that when $\beta$ is large but finite, the particle distribution in the steady state corresponds directly to the eigenvalue distribution of the $\beta$-Laguerre ensembles. We also obtained an estimation for the relaxation time in each of these cases. The asymptotic forms of the generalized Bessel function in both regimes played an important role in these estimations, and in turn, these asymptotic forms were obtained from the intertwining operator associated with the interacting Bessel process, $V_B$. We obtained an expression for $V_B$ when it operates on symmetric polynomials of squared variables.  We found that in the regimes \eref{limitbeta} and \eref{limitgamma} the behaviour of $V_B$ is almost identical to the behaviour of the intertwining operator for the interacting Brownian motions, $V_A$. This is consistent with the results for the generalized Bessel function obtained by R\"osler and Voit \cite{roslervoit08}. It seems natural to expect a similar limit behaviour in the non-symmetric case. However, there are no explicit expressions for $V_A$ or $V_B$ when they operate on non-symmetric polynomials, meaning that both the form and the limit behaviour of each of these intertwining operators remain as open problems.

We calculated the TPD of the interacting Bessel processes in the freezing regime. We found that the scaled trajectory of the particles in this regime is dependent on the square root of the Laguerre zeroes, in the same way that the freezing limit of the interacting Brownian motions depends on the Hermite zeroes. As shown in \cite{taobook12}, the Hermite zeroes are the solution of a weighted log-Fekete problem in one dimension \cite{fekete1923,deift00}. The roots of the Laguerre zeroes are also the solution to a log-Fekete problem with a different weight. Dunkl \cite{dunkl89B} studied the same extremum problem from a slightly different perspective, focusing on finding what he calls the peak set of a reflection group $W$, which is the set of unit vectors that maximizes the weight function \eref{weightk} (see \ref{generalreview}), or \eref{wkb} in the case of the root system of type $B$. This peak set and the solutions to the log-Fekete problem are closely related to the freezing limits of stochastic processes that can be expressed as radial Dunkl processes. We are currently examining this relationship in detail. In addition, from the results in this paper and in \cite{andrauskatorimiyashita12}, it is clear that the freezing limits of the interacting Brownian motions and the interacting Bessel processes are independent of the form of $V_A$ and $V_B$. This fact suggests that, in general, the form of the intertwining operator is unimportant in the freezing regime. We plan to investigate this matter in the near future.

In addition, we used numerical simulations to confirm our analytical results. We found that the particle density curves we obtained for each regime are in good agreement with the results in Theorems~\ref{freezingtypeb} and \ref{freezingoneta}; they are also in good agreement with the known exact results for $\beta=2$. While this is not surprising for the regime \eref{limitgamma}, it is somewhat surprising to see that the steady-state distribution \eref{approxbeta} fits well with the case $\beta=2$. Perhaps this means that we can obtain similar results for arbitrary values of $\beta$ or that the threshold for which $\beta$ can be considered ``large'' may be smaller than we expected. We also found that our estimations for the relaxation time in both cases are very conservative, and it seems that these estimations may have room for improvement. In particular, considering other kinds of asymptotics for the initial distribution $\mu(\bx)$ may lead to smaller lower bounds for the relaxation time.

It is known that the freezing ($\beta\to\infty$) behaviour of the interacting Brownian motions and the interacting Bessel processes is similar to the freezing trick of the Calogero-Moser (CM) systems of type $A$ and $B$, respectively \cite{polychronakos93,frahm93,yamamototsuchiya96}. It is also known \cite{forrester10,bakerforrester97,rosler08} that CM systems can be mapped to the interacting Brownian motions and other multivariate stochastic processes. In \ref{DPCMS}, we reformulate the CM systems as non-equilibrium statistical mechanics models by transforming Dunkl processes into CM systems evolving in imaginary time. For this purpose, we propose a transformation which we call the diffusion-scaling transformation. This reformulation implies that the freezing limits of these two types of systems are equivalent for all types of Dunkl processes. The freezing limit of Calogero-Moser-Sutherland (CMS, or circular CM) systems gives rise to spin chains of Haldane-Shastry type \cite{haldane88,shastry88}. This prompts the possibility of constructing a mapping from Dunkl processes defined on the unit circle to CMS systems similar to the diffusion-scaling transformation. At first glance, we expect the mapping to be fairly straightforward, because both systems are defined on a space of finite size. Hence, no diffusion scaling would be needed and only a similarity transformation should be required (perhaps a modified version of the Heckman-Opdam process is the stochastic counterpart to the CMS system \cite{schapira07}). The elliptic extension of Dyson's Brownian motion model for $\beta=2$ defined on the unit circle which was introduced in \cite{katori13} may provide hints in this direction. This is a problem that we leave open for future study. 

\ack{
The authors would like to thank the referees for their comments and suggestions, which led to the improvement of this paper. SA would like to thank T. Kimura for helpful discussions on the CM systems, N. Demni for many stimulating discussions, A. Hardy for his comments on the Fekete sets and C. F. Dunkl for helpful comments on the peak sets of reflection groups. SA would also like to thank P. Graczyk and the organizing committee of the conference ``Harmonic Analysis and Probability'' in Angers, France (September 2-8, 2012), where part of this work was carried out. 
SA is supported by the Monbukagakusho: MEXT scholarship for research students.
MK is supported by
the Grant-in-Aid for Scientific Research (C)
(Grant No.21540397) from the Japan Society for
the Promotion of Science.
}

\appendix

\section{Root systems and Dunkl processes}\label{generalreview}

\setcounter{section}{1}

We list the definitions surrounding Dunkl processes and their radial parts. For a more complete explanation, see e.g. \cite{roslervoit98, rosler08, dunklxu}.

The reflection operator is defined by
\begin{equation}\label{reflection}
\sigma_{\balpha}\bx=\bx-2\frac{\bx\cdot \balpha}{\alpha^2}\balpha.
\end{equation}
This operator acts on the vector $\bx$ by reflecting it through the hyperplane defined by the vector $\balpha$ in $N$ dimensions ($\balpha,\bx\in\RR^N$). A root system is defined as a finite set of vectors such that when its elements (called \emph{roots}) are reflected by any root, the resulting vector also belongs to the set. In other words, a root system $R$ is defined by the property that $\sigma_{\balpha}\bxi\in R$ for any $\balpha,\bxi\in R$. We denote the $i$th canonical base vector by $\be_i$. The root system of type $B$ is given by
\begin{equation}\label{RootSystemTypeB}
B_N=\{\pm \be_i : 1\leq i\leq N\}\cup\{\pm(\be_i-\be_j),\pm(\be_i+\be_j):1\leq j<i\leq N\}.
\end{equation}

For every root system $R$ there is a reflection group $W$ whose elements are the reflections along its roots and all their possible combinations. The reflections along the roots of $B_N$ are as follows:
\begin{eqnarray}
\sigma_{\pm\be_i}\bx&=&(x_1,\ldots,x_{i-1},-x_i,x_{i+1},\ldots,x_N),\nonumber\\
\sigma_{\pm(\be_i-\be_j)}\bx&=&(x_1,\ldots,x_{i-1},x_j,x_{i+1},\ldots,x_{j-1},x_i,x_{j+1},\ldots,x_N),\nonumber\\
\sigma_{\pm(\be_i+\be_j)}\bx&=&(x_1,\ldots,x_{i-1},-x_j,x_{i+1},\ldots,x_{j-1},-x_i,x_{j+1},\ldots,x_N).\label{ReflectionsTypeB}
\end{eqnarray}
Therefore, the reflection group $W_B$ contains all the permutations and sign changes that can be applied to a vector in $\RR^N$.

There is a set of parameters $k(\balpha)$, called multiplicities, that are assigned to every root system using the following rule: two roots $\balpha$ and $\bzeta$ must be associated with the same parameter, $k(\balpha)=k(\bzeta)$, if there exists an element $\rho$ of $W$ such that $\balpha=\rho\bzeta$ (i.e. they belong to the same orbit). There are two multiplicities associated with $B_N$ because the roots $\{\pm\be_i\}_{i=1}^N$ and the roots $\{\pm\be_i\pm\be_j\}_{1\leq i\neq j\leq N}$ belong to different orbits. We choose these multiplicities to be $k_1=\beta(\nu+1/2)/2$ and $k_2=\beta/2$.

Let us choose an arbitrary vector $\bmm\in\spn(R)$ such that $\balpha\cdot\bmm\neq0$ for all $\balpha\in R$. We can divide $R$ into two parts as follows: $R_+=\{\balpha\in R : \balpha\cdot\bmm>0\}$ and $R_-=\{\balpha\in R : \balpha\cdot\bmm<0\}$. Clearly, $R=R_+\cup R_-$, and both parts contain the same number of elements. These two parts are called the positive subsystem $R_+$ and the negative subsystem $R_-$. In our case, we choose the positive subsystem
\begin{equation}
B_{N,+}=\{ \be_i : 1\leq i\leq N\}\cup\{\be_i-\be_j,\be_i+\be_j:1\leq j<i\leq N\},
\end{equation}
generated, for instance, by the vector $\bmm=(1,2,\ldots,N-1,N)$. It is useful to define the sum of multiplicities over the positive subsystem as
\begin{equation}
\gamma_R=\sum_{\balpha\in R_+}k(\balpha)=\frac{1}{2}\sum_{\balpha\in R}k(\balpha).\label{gammasumk}
\end{equation}
Note that this value does not depend on the choice of $R_+$. We also introduce the following weight function:
\begin{equation}
w_k(\bx)=\prod_{\balpha\in R}|\balpha\cdot\bx|^{k(\balpha)}.\label{weightk}
\end{equation}
For the root system of type $B$, this function is given by \eref{wkb}.

For an arbitrary root system $R$, the Dunkl operator in the direction $\be_i$ is given by \cite{dunkl89, dunklxu, dunkl08}
\begin{equation}
T_i f(\bx)=\frac{\partial}{\partial x_i} f(\bx)+\sum_{\balpha \in R_+}k(\balpha)\frac{f(\bx)-f(\sigma_{\balpha} \bx)}{\balpha\cdot\bx} \alpha_i.
\end{equation}
Dunkl's intertwining operator, $V_k$, is a linear operator that is defined by the relationship
\begin{equation}
T_i V_k f(\bx)=V_k \frac{\partial}{\partial x_i} f(\bx)\label{intertwiningrelation}
\end{equation}
for any differentiable function $f(\bx)$. It is also defined so that, when it is applied on any polynomial function, the resulting polynomial has the same total degree. Also, it is normalized so that $V_k 1=1$. The intertwining operator is very powerful and useful for many calculations. However, its general explicit form is still unknown, though some progress in this direction has been achieved by Maslouhi and Youssfi \cite{maslouhiyoussfi09}. Note that $V_k$ commutes with partial derivatives with respect to time. If we denote the TPD of a Brownian motion by $p_{\textrm{B}}(t,\by|\bx)$ and we restrict $V_k$ to operate on $\bx$, we immediately have
\begin{equation}
\fl\frac{\partial}{\partial t}V_kp_{\textrm{B}}(t,\by|\bx)=V_k\frac{\partial}{\partial t}p_{\textrm{B}}(t,\by|\bx)=V_k \frac{1}{2}\Delta^{(x)}p_{\textrm{B}}(t,\by|\bx)=\frac{1}{2}\sum_{i=1}^N T_i^2 V_kp_{\textrm{B}}(t,\by|\bx).\label{dunklheatsymbolic}
\end{equation}
This means that $V_kp_{\textrm{B}}(t,\by|\bx)$ solves the Dunkl version of the diffusion equation, and $\sum_{i=1}^N T_i^2$ is called the Dunkl Laplacian. We define Dunkl processes by regarding \eref{dunklheatsymbolic} as a KBE: a process that obeys \eref{dunklheatsymbolic} as its KBE is a Dunkl process. If we denote its TPD by $P_k(t,\by|\bx)$ the KBE can be written explicitly as follows:
\begin{eqnarray}
\frac{\partial}{\partial t}P_k(t,\by|\bx)=&&\frac{1}{2}\Delta^{(x)} P_k(t,\by|\bx)+\sum_{\balpha\in R_+}k(\balpha)\frac{\balpha\cdot\bnabla^{(x)} P_k(t,\by|\bx)}{\balpha\cdot\bx}\nonumber\\
&&-\sum_{\balpha\in R_+}k(\balpha)\frac{\alpha^2}{2}\frac{P_k(t,\by|\bx)-P_k(t,\by|\sigma_{\balpha}\bx)}{(\balpha\cdot\bx)^2}.\label{dunklheat}
\end{eqnarray}
For completeness, we list the corresponding Kolmogorov forward equation (KFE): 
\begin{eqnarray}
\frac{\partial}{\partial t}P_k(t,\by|\bx)=&&\frac{1}{2}\Delta^{(y)} P_k(t,\by|\bx)-\sum_{\balpha\in R_+}k(\balpha)\frac{\balpha\cdot\bnabla^{(y)} P_k(t,\by|\bx)}{\balpha\cdot\by}\nonumber\\
&&+\sum_{\balpha\in R_+}k(\balpha)\frac{\alpha^2}{2}\frac{P_k(t,\by|\bx)+P_k(t,\sigma_{\balpha}\by|\bx)}{(\balpha\cdot\by)^2}.\label{dunklforward}
\end{eqnarray}
Radial Dunkl processes are obtained from requiring the TPD $P_k(t,\by|\bx)$ to be symmetric with respect to the reflection group $W$. More precisely, imposing the initial condition 
\begin{equation}
\mu_{\bx}(\bu)=\sum_{\rho\in W}\delta^{(N)}(\bu-\rho\bx),
\end{equation}
yields the TPD
\begin{equation}
p_k(t,\by|\bx)=\int_{\RR^N}P_k(t,\by|\bu)\mu_{\bx}(\bu)\ud\bu=\sum_{\rho\in W}P_k(t,\by|\rho\bx).\label{radialinitialcondition}
\end{equation}
Therefore, the KBE of a radial Dunkl process is given by
\begin{equation}
\frac{\partial}{\partial t}p_k(t,\by|\bx)=\frac{1}{2}\Delta^{(x)} p_k(t,\by|\bx)+\sum_{\balpha\in R_+}k(\balpha)\frac{\balpha\cdot\bnabla^{(x)} p_k(t,\by|\bx)}{\balpha\cdot\bx}.\label{KBEkRadial}
\end{equation}

For the case of the root system of type $B$, \eref{KBEkRadial} is given by \eref{typebradialdunklkbe}. Also, $p_k(t,\by|\bx)$ for the same root system is given by \eref{TPDRadialB}. For details on the root system associated with the interacting Brownian motions, $A_{N-1}$, see Appendix~A of \cite{andrauskatorimiyashita12}.

\section{Proof of Proposition~\ref{vktypeb}}\label{propositionVBproof}

We will use a slightly modified version of the proof of Theorem 2 in \cite{andrauskatorimiyashita12}, keeping in mind that $V_B$ acts only on either $\bx$ or $\by$ (we choose $\bx$ here). We begin by expanding $\sum_{\rho\in W_B}\exp(\bx\cdot\rho\by)$ in terms of symmetric polynomials. All the elements of $W_B$ can be written as compositions of variable permutations and sign changes. Then, we write
\begin{equation}
\sum_{\sigma\in W_B}\rme^{\bx\cdot\sigma\by}=\sum_{\rho\in S_N}\sum_{\substack{\mu:l(\mu)\leq N}}\frac{1}{\mu!}\sum_{\substack{\tau\in S_N:\cr \tau(\mu)\ \textrm{distinct}}}\prod_{j=1}^N\sum_{s_j=\pm 1}s_j^{\mu_{\tau(j)}}(y_{\rho(j)} x_j)^{\mu_{\tau(j)}}.
\end{equation}
The product over $j$ vanishes when at least one of the parts of $\mu$ is odd, so we only consider partitions with even parts. Then we have
\begin{eqnarray}
\sum_{\sigma\in W_B}\rme^{\bx\cdot\sigma\by}&=&\sum_{\substack{\mu:l(\mu)\leq N}}\frac{2^N}{(2\mu)!}\sum_{\substack{\tau\in S_N:\cr \tau(\mu)\ \textrm{distinct}}}\left\{\sum_{\rho\in S_N}\prod_{j=1}^N(y_{\rho(j)})^{2\mu_{\tau(j)}}\right\} \prod_{j=1}^Nx_j^{2\mu_{\tau(j)}}\nonumber\\
&=&\sum_{\substack{\mu:l(\mu)\leq N}}\frac{2^NN!}{(2\mu)!}\frac{m_{\mu}[(\bx)^2]m_{\mu}[(\by)^2]}{M(\mu,N)}.\label{exponentialexpansion}
\end{eqnarray}
Applying $V_B$ on this result and inserting into \eref{genbesseltypeb} yields
\begin{equation}
\fl V_B\sum_{\substack{\mu:l(\mu)\leq N}}\frac{m_{\mu}[(\bx)^2]m_{\mu}[(\by)^2]}{(2\mu)!M(\mu,N)}=\FF{2/\beta}\left(\frac{\beta}{2}(\nu+N-1/2) +\frac{1}{2};\frac{(\bx)^2}{2},\frac{(\by)^2}{2}\right).
\end{equation}
From \eref{genhypergeob} and the fact that Jack polynomials are homogeneous ($\PP{\tau}{\alpha}(c\bx)=\PP{\tau}{\alpha}(\bx)c^{|\tau|}$), we obtain
\begin{equation}
\fl V_B\!\!\!\!\sum_{\substack{\mu:l(\mu)\leq N}}\!\!\!\frac{m_{\mu}[(\bx)^2]m_{\mu}[(\by)^2]}{(2\mu)!M(\mu,N)}=\!\!\!\!\sum_{\tau:l(\tau)\leq N}\frac{c_\tau (2/\beta)}{c_\tau^\prime (2/\beta)}\frac{\PP{\tau}{2/\beta}[(\bx)^2]\PP{\tau}{2/\beta}[(\by)^2]}{2^{2|\tau|}(\frac{\beta}{2}[\nu+N-1/2]+\frac{1}{2})_\tau^{(2/\beta)}(N\frac{\beta}{2})_\tau^{(2/\beta)}}.
\end{equation}
Next, we use the inverse of the expansion of Jack polynomials into monomial symmetric polynomials, \eref{JackP}, on the lhs:
\begin{equation}
\fl \sum_{\substack{\mu:l(\mu)\leq N}}\!\!\frac{V_Bm_{\mu}[(\bx)^2]}{(2\mu)!M(\mu,N)}m_{\mu}[(\by)^2]=\!\!\!\!\sum_{\substack{\mu:l(\mu)\leq N}}\!\!\frac{V_Bm_{\mu}[(\bx)^2]}{(2\mu)!M(\mu,N)}\sum_{\substack{\rho:\rho\leq\mu\cr |\rho|=|\mu|}}(u^{-1})_{\mu\rho}(2/\beta)\PP{\rho}{2/\beta}[(\by)^2].
\end{equation}
Because Jack polynomials are orthogonal, we equate the coefficients of $\PP{\tau}{2/\beta}[(\by)^2]$,
\begin{equation}
\fl \sum_{\substack{\mu:l(\mu)\leq N\cr |\mu|=|\tau|}}\!\!\frac{V_Bm_{\mu}[(\bx)^2]}{(2\mu)!M(\mu,N)}(u^{-1})_{\mu\tau}(2/\beta)=\frac{c_\tau (2/\beta)}{c_\tau^\prime (2/\beta)}\frac{\PP{\tau}{2/\beta}[(\bx)^2/4]}{(\frac{\beta}{2}[\nu+N-1/2]+\frac{1}{2})_\tau^{(2/\beta)}(N\frac{\beta}{2})_\tau^{(2/\beta)}}.
\end{equation}
Multiplying by $\sum_\tau u_{\tau\lambda}(2/\beta)$ on both sides completes the proof.

\section{The Roots of the Laguerre Polynomials}\label{laguerreappendix}

In this appendix, we prove that the solutions of \eref{minconditionb} are located at the square root of the zeroes of the $N$th Laguerre polynomial, and we prove that the minimum value of $F(\bz,\nu,N)$, given by \eref{functionFforb}, is zero. First, we consider the second derivatives of $F(\bz,\nu,N)$:
\begin{equation}
\fl\frac{\partial^2}{\partial z_j \partial z_i}F(\bz,\nu,N)=2\delta_{ij}\Big(1+\frac{2\nu+1}{2z_i^2}\Big)+4\Big[\delta_{ij}\sum_{\substack{l:l\neq i\cr l=1}}^N\frac{z_i^2+z_l^2}{(z_i^2-z_l^2)^2}-(1-\delta_{ij})\frac{2z_iz_j}{(z_i^2-z_j^2)^2}\Big].
\end{equation}
If we consider an arbitrary vector $\bu\in\RR^N$, we see that
\begin{eqnarray}
\fl\sum_{1\leq i,j\leq N}\!\!\!\!u_iu_j\frac{\partial^2}{\partial z_j \partial z_i}F(\bz,\nu,N)&=&2\sum_{i=1}^Nu_i^2\Big(1+\frac{2\nu+1}{2z_i^2}\Big)\nonumber\\
&&+2\!\!\!\!\sum_{1\leq i\neq j\leq N}\!\!\!\!\frac{(u_iz_i-u_jz_j)^2+(u_iz_j-u_jz_i)^2}{(z_i^2-z_j^2)^2}\geq 0.\label{HessianFPositiveDefinite}
\end{eqnarray}
Therefore, all extrema are minima. Now, we substitute $\bor=(\bz)^2$ in \eref{minconditionb} to obtain
\begin{equation}\label{minimizationonr}
r_i=\nu+1/2+\sum_{\substack{j:j\neq i\cr j=1}}^N\frac{2r_i}{r_i-r_j},\quad i=1,\ldots,N.
\end{equation}
We construct the following polynomial,
\begin{equation}
p_N(x)=C_N\prod_{j=1}^N(x-r_j),
\end{equation}
where $C_N$ is an arbitrary real nonzero constant, and we denote by $p_N^\prime(x)$ and $p_N^{\prime\prime}(x)$ its first and second derivatives, respectively. Evaluating them at $x=r_i$, they become
\begin{eqnarray}
p_N^\prime(r_i)&=&C_N\prod_{\substack{n:n\neq i\cr n=1}}^N(r_i-r_n)\ \textrm{and}\\
p_N^{\prime\prime}(r_i)&=&2C_N\sum_{\substack{j:j\neq i\cr j=1}}^N\prod_{\substack{n:n\neq i,j\cr n=1}}^N(r_i-r_n).
\end{eqnarray}
Multiplying \eref{minimizationonr} by $p_N^\prime(r_i)$ yields
\begin{equation}
r_ip_N^{\prime\prime}(r_i)+(\nu+1/2 -r_i)p_N^\prime(r_i)=0,\quad i=1,\ldots,N.
\end{equation}
Comparing the above with the equation obeyed by the Laguerre polynomials \cite{szego},
\begin{equation}\label{LaguerreDifferentialEquation}
xL_N^{(\alpha)\prime\prime}(x)+(\alpha+1-x)L_N^{(\alpha)\prime}(x)+NL_N^{(\alpha)}(x)=0,
\end{equation}
we see that $p_N(x)$ must be proportional to $L_N^{(\nu-1/2)}(x)$, and the $\{r_i\}_{i=1}^N$ must be the roots of $L_N^{(\nu-1/2)}(x)$, $\{s_{i,\nu-1/2}\}_{i=1}^N$. This, in turn, means that the minima of $F(\bz,\alpha+1/2,N)$ lie at $\bz=(\sqrt{s_{1,\alpha}},\ldots,\sqrt{s_{N,\alpha}})$.

We proceed to calculate the minimum value of $F(\bz,\alpha+1/2,N)$, given by
\begin{eqnarray}
\fl \min_{\bz\in\RR^N} F(\bz,\alpha+1/2,N)&=&\sum_{i=1}^Ns_{i,\alpha}-(\alpha+1)\sum_{i=1}^N\log s_{i,\alpha}-2\sum_{1\leq i<j\leq N}\log|s_{j,\alpha}-s_{i,\alpha}| \nonumber\\
&&+\sum_{i=1}^N i\log i+\sum_{i=1}^N(\alpha+i)\log(\alpha+i)-N(\alpha+N).\label{minimumvalue}
\end{eqnarray}
The first term is obtained from the sum over $i$ of \eref{minimizationonr},
\begin{equation}
\sum_{i=1}^Ns_{i,\alpha}=N(\alpha+N).\label{partone}
\end{equation}
The second term can be calculated from
\begin{equation}
\sum_{i=1}^N\log s_{i,\alpha}=\log(N!L_N^{(\alpha)}(0)),
\end{equation}
where $L_N^{(\alpha)}(0)=\frac{1}{N!}\prod_{i=1}^N(\alpha+i)$ \cite{szego}, so we obtain
\begin{equation}
\sum_{i=1}^N\log s_{i,\alpha}=\sum_{i=1}^N\log(\alpha+i).\label{parttwo}
\end{equation}
Finally, following \cite{szego} and in a manner similar to \cite{andrauskatorimiyashita12}, we calculate the third term. To avoid confusions, we will denote the roots of $L_{N}^{(\alpha)}(x)$ by $\{s_{i,N,\alpha}\}_{i=1}^N$. We write
\begin{equation}
\prod_{1\leq i<j\leq N}(s_{j,N,\alpha}-s_{i,N,\alpha})^2=(-1)^{N(N+1)/2}(N!)^N\prod_{i=1}^NL_{N}^{(\alpha)\prime}(s_{i,N,\alpha}).
\end{equation}
Using the derivative relation $xL_{N}^{(\alpha)\prime}(x)=NL_{N}^{(\alpha)}(x)-(N+\alpha)L_{N-1}^{(\alpha)}(x)$ combined with \eref{parttwo}, we have
\begin{equation}
\fl\prod_{1\leq i<j\leq N}(s_{j,N,\alpha}-s_{i,N,\alpha})^2=\frac{(-1)^{N(N-1)/2}(N!)^N(\alpha+N)^N}{\prod_{j=1}^N(\alpha+j)}\prod_{i=1}^NL_{N-1}^{(\alpha)}(s_{i,N,\alpha}).\label{tobeinserted}
\end{equation}
The product of Laguerre polynomials can be rewritten as
\begin{equation}
\prod_{i=1}^NL_{N-1}^{(\alpha)}(s_{i,N,\alpha})=\frac{N^N}{N!}\prod_{i=1}^{N-1}L_{N}^{(\alpha)}(s_{i,N-1,\alpha}),
\end{equation}
and we use the recursion relation $NL_{N}^{(\alpha)}(x)=(-x+2N+\alpha-1)L_{N-1}^{(\alpha)}(x)-(N+\alpha-1)L_{N-2}^{(\alpha)}(x)$ on the above to obtain
\begin{equation}
\prod_{i=1}^NL_{N-1}^{(\alpha)}(s_{i,N,\alpha})=\frac{(-1)^{N-1}(N-1+\alpha)^{N-1}}{(N-1)!}\prod_{i=1}^{N-1}L_{(N-1)-1}^{(\alpha)}(s_{i,N-1,\alpha}).
\end{equation}
Mathematical induction on this expression yields
\begin{equation}
\prod_{i=1}^NL_{N-1}^{(\alpha)}(s_{i,N,\alpha})=(-1)^{N(N-1)/2}\prod_{i=1}^{N-1}\left(\frac{\alpha+i}{N-i}\right)^i.
\end{equation}
Inserting this expression into \eref{tobeinserted} and taking logarithms on both sides gives
\begin{equation}
2\sum_{1\leq i<j\leq N}\log|s_{j,N,\alpha}-s_{i,N,\alpha}|=\sum_{i=1}^N[(i-1)\log(\alpha+i)+i\log i].\label{partthree}
\end{equation}
Substituting \eref{partone}, \eref{parttwo} and \eref{partthree} into \eref{minimumvalue} gives $\min_{\bz\in\RR^N} F(\bz,\alpha+1/2,N)=0$.

\section{Correspondence Between Dunkl Processes and Calogero-Moser Systems}\label{DPCMS}

Under an arbitrary root system $R$, the CM systems on a line with a harmonic background potential and an inverse-square repulsion potential are given by the Hamiltonian
\begin{equation}
\mathcal{H}_{\text{CM}}^R=-\frac{1}{2}\Delta^{(x)}+\sum_{\balpha\in R_+}\frac{\alpha^2}{2}\frac{k(\balpha)[k(\balpha)-\sigma_{\balpha}]}{(\balpha\cdot\bx)^2}+\frac{\omega^2}{2}\sum_{i=1}^Nx_i^2,\label{generalcm}
\end{equation}
where all the particles have been chosen to be of unit mass, and we have taken $\hbar =1$. This Hamiltonian corresponds to Equation (2.9) in \cite{khastgir00} with the inverse-square potential $V(r)=1/r^2$. (See \ref{generalreview} for the definition of the expressions used here.)

Dunkl operators have been used as a tool to prove the integrability of the CM systems (see \cite{forrester10}, Sections~11.4.2 to 11.5.3). It has been shown under several root systems \cite{rosler98} that after applying a similarity transformation (using the ground state eigenfunction), the CM system Hamiltonian is expressed as a Dunkl Laplacian plus a restoring term, $-\bx\cdot\bnabla^{(x)}$ (here, $\bnabla^{(x)}=(\frac{\partial}{\partial x_1},\ldots,\frac{\partial}{\partial x_N})$). One can then find the polynomial eigenfunctions for the transformed Hamiltonian as stated in \cite{bakerdunklforrester} and shown in \cite{bakerforrester97}. We aim to transform the KFE of a Dunkl process into the Schr{\"o}dinger equation of the CM systems. 

We define what we call the diffusion-scaling transformation as follows. In view of the transformation of a simple Brownian motion into a one-dimensional quantum harmonic oscillator in imaginary time proposed in \cite{katoritanemura07}, we consider the substitution given by
\begin{equation}
(t,\by)\to(\tau,\bzeta)=\Bigg(\frac{\ln t}{2\omega},\frac{\by}{\sqrt{2\omega t}}\Bigg).\label{substitutionr}
\end{equation}
We denote the density of the Dunkl process at a time $t$ for a given initial configuration by $f(t,\by)$. We perform the variable substitution $u[t(\tau,\bzeta,\omega),\by(\tau,\bzeta,\omega)]=u(\tau,\bzeta)$ and after that we apply the similarity transformation
\begin{equation}
u(\tau,\bzeta)=\exp[-W(\tau,\bzeta)]U(\tau,\bzeta)\label{transformationr}
\end{equation}
with $W(\tau,\bzeta)$ given by
\begin{equation}
W(\tau,\bzeta)=\frac{1}{2}\omega\sum_{i=1}^N\zeta_i^2-\frac{1}{2}\ln w_k(\bzeta)+\omega N\tau.\label{Wr}
\end{equation}
Because the scaling is isotropic, it is independent of the root system $R$.

\begin{proposition}\label{correspondencer}
The diffusion-scaling transformation given by \eref{substitutionr}, \eref{transformationr}, and \eref{Wr} transforms the Dunkl process on the root system $R$ into the CM system with harmonic confinement on the same root system evolving in imaginary time.
\end{proposition}
\begin{proof}
We transform the KFE \eref{dunklforward}. We begin by considering a Dunkl process whose distribution is given by $f(t,\by)$ for some initial condition. We write down the derivatives in time and space in terms of the new variables as follows:
\begin{eqnarray}
\frac{\partial}{\partial t}&=&\frac{1}{2\omega t}\frac{\partial}{\partial \tau}-\frac{1}{2t}\bzeta\cdot\bnabla^{(\zeta)},\nonumber\\
\frac{\partial}{\partial y_i}&=&\frac{1}{\sqrt{2\omega t}}\frac{\partial}{\partial \zeta_i}.\label{dertransformation1}
\end{eqnarray}
The differential operators that result from inserting the above in \eref{dunklforward} are transformed by \eref{transformationr} as follows:
\begin{eqnarray}
\e^{W}\frac{\partial}{\partial \tau}\e^{-W}&=&\frac{\partial}{\partial \tau}-\omega N,\nonumber\\
\e^{W}\frac{\partial}{\partial \zeta_i}\e^{-W}&=&\frac{\partial}{\partial \zeta_i}-\omega\zeta_i+\sum_{\balpha\in R_+}\frac{k(\balpha)}{\balpha\cdot\bzeta}\alpha_i,\nonumber\\
\e^{W}\Delta^{(\zeta)}\e^{-W}&=&\Delta^{(\zeta)}+2\Bigg(\sum_{\balpha\in R_+}\frac{k(\balpha)}{\balpha\cdot\bzeta}\balpha-\omega\bzeta\Bigg)\cdot\bnabla^{(\zeta)}+\omega^2\zeta^2-(2\gamma_R+N)\omega\nonumber\\
&&+\sum_{\balpha\in R_+}\sum_{\bxi\in R_+}\frac{k(\balpha)k(\bxi)}{(\balpha\cdot\bzeta)(\bxi\cdot\bzeta)}\balpha\cdot\bxi-\sum_{\balpha\in R_+}\frac{k(\balpha)}{(\balpha\cdot\bzeta)^2}\alpha^2.\label{differentialtransformationsr}
\end{eqnarray}
Therefore, inserting \eref{dertransformation1} and \eref{differentialtransformationsr} successively in \eref{dunklforward} yields
\begin{eqnarray}
\frac{\partial}{\partial \tau}U(\tau,\bzeta)&=&\frac{1}{2}\Delta^{(\zeta)}U(\tau,\bzeta)+\frac{\omega}{2}[2\gamma_R+N-\omega\zeta^2]U(\tau,\bzeta)\nonumber\\
&&+\sum_{\balpha\in R_+}\frac{\alpha^2}{2}\frac{k(\balpha)}{(\balpha\cdot\bzeta)^2}U(\tau,\sigma_{\balpha}\bzeta)\nonumber\\
&&\ -\sum_{\balpha\in R_+}\sum_{\bxi\in R_+}\frac{\balpha\cdot\bxi}{2}\frac{k(\balpha)k(\bxi)}{(\balpha\cdot\bzeta)(\bxi\cdot\bzeta)}U(\tau,\bzeta).\label{withdoublesum}
\end{eqnarray}
The double sum in the bottom term of the equation above can be simplified because all the terms where $\balpha\neq\bxi$ cancel each other (see Lemma~4.4.6 of \cite{dunklxu}). By denoting the ground-state energy by $E_{\text{0}}^{R}=\omega(\gamma_R+N/2)$ and using $\HH{CM}^{R}$ with $\bzeta$ instead of $\bx$, we finally obtain
\begin{equation}
-\frac{\partial}{\partial \tau}U(\tau,\bzeta)=[\HH{CM}^{R}-E_{\text{0}}^{R}]U(\tau,\bzeta).\qedhere\label{done}
\end{equation}
\end{proof}


{\it Remark:} this proof involves only straightforward calculations, with the notable exception of the step required to simplify the double sum in \eref{withdoublesum}. This is perhaps the most important part of the proof, and it is not trivial. The simplest case is when $R$ is the root system of type $A$ (see, e.g., \cite{forrester10}, Proposition~11.3.1). Note also that Proposition~\ref{correspondencer} only requires that $\omega>0$. If $\omega=0$, there is no need to use the diffusion scaling \eref{substitutionr}, and one may simply apply a similarity transformation on the Dunkl process to obtain the unconfined CM system on the same root system. 
\\
\bibliography{biblio}

\end{document}